\documentclass[a4paperb,11pt]{article}
\usepackage[utf8]{inputenc}
\usepackage[british]{babel}
\usepackage[mono=false]{libertine}
\usepackage[T1]{fontenc}
\usepackage{amsthm}
\usepackage[top=2.5cm, bottom=2.5cm, left=2cm, right=2cm]{geometry}
\usepackage{xcolor}
\usepackage[font=small]{caption}
\usepackage{cite}
\usepackage[pdftex]{graphicx}
\DeclareGraphicsExtensions{.pdfsx,.jpeg,.png}
\usepackage[cmex10]{amsmath}
\usepackage{array}
\usepackage{url} 
\usepackage{mathtools}
\usepackage{amssymb,amsfonts}
\usepackage{dsfont}
\usepackage{stmaryrd}
\usepackage{bm}
\usepackage{citesort}
\usepackage{color}
\usepackage{titlesec}
\usepackage{enumitem}
\makeatletter
\newcommand{\setappendix}{Appendix~\thesection:~~}
\newcommand{\setsection}{\thesection~~}
\titleformat{\section}{\bfseries\LARGE}{%
	\ifnum\pdfstrcmp{\@currenvir}{appendices}=0
	\setappendix
	\else
	\setsection
\fi}{0em}{}
\makeatother
\usepackage[titletoc]{appendix}
\usepackage[pdfpagemode=UseNone,bookmarksopen=false,colorlinks=true,urlcolor=blue,citecolor=magenta,citebordercolor=blue,linkcolor=blue]{hyperref}
\usepackage{setspace}
\usepackage{notations}
\newcommand{\<}{\langle}
\renewcommand{\>}{\rangle}

\newcommand{\sH}{\mathcal{H}}
\newcommand{\sR}{\mathcal{R}}
\newcommand{\sX}{\mathcal{X}}
\newcommand{\sL}{\mathcal{L}}

\newcommand{\tX}{x}
\newcommand{\tbX}{\bm{x}}

\newtheorem{theorem}{Theorem}[section]
\newtheorem{lemma}[theorem]{\textbf{Lemma}}
\newtheorem{thm}[theorem]{\textbf{Theorem}}

\newtheorem{proposition}[theorem]{\textbf{Proposition}}

\DeclareMathAlphabet{\varmathbb}{U}{bbold}{m}{n}

\newcommand{\EE}{\mathbb{E}}

\begin{document}
\title{Mutual information for the stochastic block model\\by the adaptive interpolation method}
\author{Jean Barbier$^{*}$, Chun Lam Chan$^{\dagger}$, and Nicolas Macris$^{\dagger}$}
\date{}
\maketitle
{\let\thefootnote\relax\footnote{
\hspace{-18.5pt}
$*$ The Abdus Salam International Center for Theoretical Physics, Trieste, Italy.\\
$\dagger$ Communication Theory Laboratory, \'Ecole Polytechnique F\'ed\'erale de Lausanne, Switzerland.
}}
\begin{abstract}
We rigorously derive a single-letter variational expression for the mutual information of the asymmetric two-groups stochastic 
block model in the dense graph regime. Existing proofs in the literature are indirect, as they involve mapping the model 
to a rank-one matrix estimation problem whose mutual information is then determined by a combination of 
methods (e.g., interpolation, cavity, algorithmic, spatial coupling). In this contribution we provide a self-contained and direct proof using only the recently introduced adaptive interpolation method.
\end{abstract}
\section{Introduction}

The stochastic block model (SBM) has a long history and has attracted the attention of many disciplines. It was first introduced as a model of 
community detection in the networks and statistics 
literature \cite{HHL:1983}, as a problem of finding graph bisections in theoretical computer science \cite{DBLP:conf/focs/BuiCLS84}, and has also 
been proposed as a model for inhomogeneous random graphs \cite{PhysRevE.66.066121,Bollobas:2007:PTI:1276871.1276872}. Here we adopt the 
community detection interpretation and motivation
\cite{Fortunato:2010}.
A partition of nodes into labeled groups is hidden to an observer who is only given a random graph generated on the basis of the partition. The task of the observer is to recover the hidden partition from the observed graph. A simple setting that lends itself to mathematical analysis is the following. The labels of nodes are drawn i.i.d. from a prior distribution and, for the graph, the edges between pairs of nodes are placed independently according to a probability which depends only on the 
group labels. If the probability is slightly higher (resp. lower) when the pair of nodes have the same label the model is called assortative (resp. disassortative). 
Moreover we suppose that the parameters of the prior and edge probability distributions are all known so that we are working in the framework of Bayesian (optimal) inference.
Note that the recovery task is non-trivial only when parameters are such that no information about the group label is revealed from the degrees of nodes. 
Much progress has been done in recent years within this simple mathematical setting and 
we refer to \cite{abbe2018community} for a recent comprehensive review and references. 
 
In the limit of large number of nodes the SBM displays interesting phase transitions for (partial) recovery of the hidden partition and much effort has been deployed to characterize the phase diagram, in terms of information theoretic as well as algorithmic phase transition thresholds, and compute the algorithmic-to-statistical gaps. In this vein a fundamental quantity 
 is the {\it mutual information} between the hidden labels of the nodes and the observed graph. Indeed from the asymptotic value of the mutual information per node one can compute information theoretic thresholds of recovery. 
In this paper we focus on the mutual information of the {\it two-group} SBM with possibly {\it asymmetric} group sizes, in {\it dense regimes} where  the expected 
degree of the nodes diverges with the total number of nodes (and is independent of the group label). 
We rigorously determine a {\it single-letter variational expression} for the asymptotic mutual information
by means of the recently developed {\it adaptive interpolation method} \cite{BarM:2018,2019arXiv190106516B}. 

Single-letter variational expressions for the mutual information of the SBM are not new. They were first analytically derived in heuristic ways 
by methods of statistical physics and in this context are often called replica or cavity formulas \cite{LesKZ:2015}. Rigorous proofs then 
appeared in \cite{DesAM:2017,LeiM:2018}. These approaches are indirect in the sense that the SBM is 
first mapped on a rank-one matrix factorization problem, and then the matrix factorization 
problem is solved. In \cite{DesAM:2017} the particular case of two equal size 
communities is considered and the analysis relies on the fact that in this case the information theoretic phase 
transition is of the second order type (i.e., continuous) which allows to use message-passing arguments. 
The asymmetric case is more challenging because first order (discontinuous) phase transitions 
appears for large enough asymmetry. In \cite{LeiM:2018} this case is tackled through a Guerra-Toninelli interpolation combined with a rigorous version of 
the cavity method or Aizenman-Sims-Starr scheme \cite{AizSS:2003}. Strictly speaking the analysis \cite{LeiM:2018} does not cover the widest possible regime of dense graphs (see section two for details).
We note that the mutual information of rank-one matrix factorization had also 
been determined earlier in \cite{KorM:2009} for the symmetric case and more recently for the 
general case in \cite{XXT,DBLP:journals/corr/abs-1812-02537} using a spatial coupling method.
 
The proof presented here covers the asymmetric two-group SBM and has 
the virtue of being completely unified. It uses a single method, namely the adaptive 
interpolation, is conceptually simpler, and is direct as it does not make any 
detour through another model. The method is a powerful evolution of the classic 
Guerra-Toninelli interpolation \cite{guerra2002thermodynamic} and allows to derive tight 
upper {\it and} lower bounds for the mutual information, whereas the classic interpolation only 
yields a one-sided inequality. It has been successfully applied to a range of Bayesian inference 
problems, e.g., \cite{2017arXiv170910368B,pmlr-BarKMMZ:2018}. Here, besides various new technical 
aspects, the main novelty is that we do not use Gaussian integration by parts, as is generally the case in interpolation methods. Instead, we develop 
a general {\it approximate integration by parts} formula and apply it to the Bernoulli random elements of the adjacency 
matrix of the graph. We note that related approximate integration by parts formulas have already been 
used by \cite{Talagrand:2004,CarH:2006} in the context of the Hopfield and Sherrington-Kirkpatrick models.

It would be desirable to extend the present method to the sparse regime of the SBM where the average degree of the nodes stays finite as the number of nodes diverges.
This is much more challenging however, and the mutual information has so far been determined only for the disassortative case \cite{Coja-Oghlan:2017} while the assortative case remains open. The thresholds however have been successfully determined for both cases in \cite{DecKMZ:2011,Mossel2018,Krzakala2013SpectralRC,BorLM:2015}. The adaptive interpolation method has been developed for the related censored block model in the sparse regime \cite{BarCM:2018} and hopefully it can be also extended to the sparse SBM, which we leave for future work.


%
%

%
\section{Setting and results: asymmetric two-groups SBM}
We first formulate the SBM for two communities that may be of different sizes. 
Suppose we have $n$ nodes belonging to two communities where the partition is denoted by a vector $\bm{X}^0 \in \{-1, 1\}^n$. 
Labels $X_i^0$ are i.i.d. Bernoulli random variables with $\mathbb{P}(X_i^0 = 1) = r \in (0,1/2]$. The size of each community is $nr$ and $n(1-r)$ up to fluctuations 
of ${\cal O}(\sqrt n)$. The labels $\bm{X}^0$ are hidden and instead one is given a random undirected graph $\bm{G}$ constructed as follows (equivalently one is given an adjacency marix). An edge between node $i$ and $j$ is present with probability $\mathbb{P}(G_{ij}=1 | X_i^0,X_j^0)$ and absent with the complementary probability.
To specificy $\mathbb{P}(G_{ij}=1 | X_i^0, X_j^0)$, first we define $d_n$ such that
\begin{align}
 \mathbb{E}[\mathrm{deg}(i)  | X_i^0 = 1] &\equiv \frac{(n-1)d_n}{n} \approx d_n\,, \label{eq:con1} \\
 \mathbb{E}[\mathrm{deg}(i)  | X_i^0 = -1] &\equiv \frac{(n-1)d_n}{n} \approx d_n\,. \label{eq:con2}
\end{align}
We require these two constraints for the inference problem to be non-trivial, in the sense that no information about the labels stems from the nodes' degrees. The two constraints imply
\begin{align*}
\mathbb{E}[\mathrm{deg}(i)] =  r \ \mathbb{E}[\mathrm{deg}(i)  | X_i^0 = 1] + (1-r) \mathbb{E}[\mathrm{deg}(i)  | X_i^0 = -1] = \frac{(n-1)d_n}{n} \approx d_n
\end{align*}
so that we can interpret $d_n$ as the average degree of a node.
Then we define $\mathbb{P}(G_{ij}=1 | X_i^0, X_j^0) = M_{X_i^0, X_j^0}$
where $M_{X_i^0, X_j^0}$ are the four possible matrix elements of 
$$ M = \frac{d_n}{n} \begin{bmatrix} a_n & b_n \\ b_n & c_n \end{bmatrix}.$$
Because of \eqref{eq:con1} and \eqref{eq:con2}, we have the equations
\begin{align*}
 \mathbb{E}[\mathrm{deg}(i)  | X_i^0 = 1]&=\frac{(n-1)d_n}{n} (ra_n+(1-r)b_n) = \frac{(n-1)d_n}{n}\,, \\
 \mathbb{E}[\mathrm{deg}(i)  | X_i^0 =-1] &=\frac{(n-1)d_n}{n} (rb_n+(1-r)c_n) = \frac{(n-1)d_n}{n}\,.
\end{align*}
Solving this system imposes 
$a_n = 1 - (1 -1/r)(1-b_n)$ and $c_n = 1 - (1-b_n)/(1-1/r)$. Therefore there are three independent parameters, namely $d_n$, $b_n$ and $r$.
A more convenient re-parametrization is often used \cite{DesAM:2017} instead of $b_n, d_n$: 
$$
\bar{p}_n \equiv\frac{d_n}{n}\,, \qquad \text{and}\qquad \Delta_n \equiv \frac{d_n(1-b_n)}{n}\,.
$$ 
Here $\bar{p}_n\in (0,1)$ is the average probability for the presence of an edge. We will look at the {\it dense} asymmetric SBM (the symmetric model corresponding to $r=1/2$) regimes where 
$d_n = n\bar{p}_n\to +\infty$. In our analysis the growth of $d_n$ spans the whole spectrum from arbitrarily slow, at the verge of a sparse graph, to  linear $d_n = vn$, $v\in (0,1)$, for fully dense graphs.

In this paper we rigorously determine the asymptotic mutual information for this problem $\lim_{n\to \infty}\frac1n I(\bm{X}^0;\bm{G})$ in the {\it dense graph regime} wherein $\bar{p}_n$ and $\Delta_n$ satisfy:
\begin{enumerate}[label=(h\arabic*),noitemsep]
\item \label{h1} (Dense SBM) $n\bar{p}_n(1-\bar{p}_n)^3 \xrightarrow{n \rightarrow \infty} \infty$. 

\item \label{h2} (Appropriate scaling of signal-to-noise ratio) $\lambda_n \equiv n\Delta_n^2 / \big ( \bar{p}_n (1-\bar{p}_n) \big ) =d_n(1-b_n)^2 / (1-d_n/n) \xrightarrow{n \rightarrow \infty} \lambda$ finite.
\end{enumerate}

The first condition ensures that the graph is dense in the sense that $d_n\to +\infty$, still maintaining $\bar{p}_n\in (0,1)$.
The second ensures the mutual information has a well defined non-trivial limit when $n\to +\infty$. 
Note that the second condition requires $\Delta_n \ll \bar{p}_n(1-\bar{p}_n)^2$ as $\Delta_n / \big ( \bar{p}_n (1-\bar{p}_n)^2 \big ) 
= \sqrt{\lambda_n / (n \bar{p}_n (1-\bar{p}_n)^3)} \rightarrow 0$ as $n \rightarrow \infty$, hence $\Delta_n \ll \bar{p}_n$ and $\Delta_n \ll (1-\bar{p}_n)^2$. 
The reader may wish to keep in mind two simple typical examples. The first example is a dense graph with $d_n = vn$, $v\in [0,1]$ so $\bar{p}_n = v$ and $\Delta_n \approx \sqrt{\lambda v(1-v)/ n}$. 
The second example is $d_n=v n^{1-\theta}$ with $\theta\in (0,1)$, so $p_n= vn^{-\theta}$ and $\Delta_n \approx \sqrt{\lambda v n^{-1-\theta}}$. These are easily translated back to the matrix $M$. 

We note that in the sparse graph version of the model one would have a finite limit for $d_n$ but 
the second condition would be the same.
The analysis of the sparse case is however more difficult and is not addressed in this paper. 

Instead of working with the Ising spin $\pm1$ variables it is convenient to change the alphabet. We define 
${X}_i \equiv \phi_{r}(X_i^0)$ with $\phi_{r}(1) = \sqrt{(1-r)/r}$ and $\phi_{r}(-1) = -\sqrt{r/(1-r)}$. The hidden labels of the nodes now belong to the alphabet
$\sX \equiv \{\sX_1=\sqrt{(1-r)/r},\sX_2=-\sqrt{r/(1-r)}\}$ and ${\bm{X}} \in \sX^n$.  An edge is then present with conditional probability
\begin{align}
\mathbb{P}(G_{ij}=1 | {X}_i {X}_j) = \bar{p}_n + \Delta_n {X}_i {X}_j\,. \label{eq:transition-prob}
\end{align}
This can be viewed as an asymmetric binary-input binary-output channel $\bm{X} \to \bm{G}$ and 
the inference problem is to recover the input $\bm{X}$ (or $\bm{X}^0$) from the channel output $\bm{G}$. Henceforth we adopt the notation 
$$
\mathbb{P}_r\equiv r\delta_{\sX_1}+(1-r)\delta_{\sX_2}
$$ 
for the probability distribution of the hidden labels $\sX\in\mathcal{X}$. Note that $\mathbb{E}[X^2] =1$.

We now formulate our results which provide a single-letter 
variational formula for the asymptotic mutual information. Let $Z\sim\mathcal{N}(0,1)$ and $X\sim \mathbb{P}_r$ independently, and set for $q >0$:
\begin{align*}
\Psi(q, \lambda, r)
	& \equiv \frac{\lambda}{4} + \frac{q^2}{4 \lambda}  - \mathbb{E} \ln  \sum_{x\in\sX} \mathbb{P}_r(x) e^{\sqrt{q}\,Z x + q {X} x - \frac{q}{2} x^2 } \, .
\end{align*}
The so-called replica formula conjectures the identity
\begin{align}\label{replicaformula}
\lim_{n \rightarrow \infty} \frac{1}{n} I(\bm{X}^0; \bm{G}) = \min_{q \in [0, \lambda]} \Psi(q, \lambda, r)\,.
\end{align}
We prove that \eqref{replicaformula} is correct, namely:
\begin{thm}[Upper bound]\label{thm:ub}
For the SBM under concern in the regime \ref{h1}, \ref{h2},
\begin{align*}
{\textstyle \limsup_{n \rightarrow \infty} \frac{1}{n} I(\bm{X}^0; \bm{G}) \leq \min_{q \in [0,\lambda]} \Psi(q, \lambda, r)}\,.
\end{align*}
\end{thm}
\begin{thm}[Lower bound]\label{thm:lb}
For the SBM under concern in the regime \ref{h1}, \ref{h2},
\begin{align*}
{\textstyle \liminf_{n \rightarrow \infty} \frac{1}{n} I(\bm{X}^0; \bm{G}) \geq \min_{q \in [0,\lambda]} \Psi(q, \lambda, r)}\,.
\end{align*}
\end{thm}

{\it Remark 1:} Of course we have $I(\bm{X}^0; \bm{G}) = I(\bm{X}; \bm{G})$ and in the following we will work with 
$I(\bm{X}; \bm{G})$ where $\bm{X} \in \sX = \{\sX_1=\sqrt{(1-r)/r},\sX_2=-\sqrt{r/(1-r)}\}$.

{\it Remark 2:} Elementary analysis shows that the minimum over $q\geq 0$ of $\Psi(q,\lambda,r)$ is attained for $q\in [0,\lambda]$. 

{\it Remark 3:} From \eqref{replicaformula} one can derive the information theoretic phase transition thresholds.
Let $r_* \equiv (1-1/\sqrt 3)/2$. 
For "small" asymmetry between group sizes $r\in [r_*, 1/2]$ there is a continuous phase transition at $\lambda_c =1$ while for "large" asymmetry $r\in \ ]0, r_*[$ the phase transition becomes discontinuous. An information theoretic-to-algorithmic gap occurs in the second situation as discussed in detail in \cite{LeiM:2018}.

Let us explain the relation of these theorems with previous works. In \cite{DesAM:2017} they were obtained for the symmetric case $r=1/2$ by a mapping of the model
on a rank-one matrix estimation problem via an application of Lindeberg's theorem. The regime treated is essentially the same than ours except that in place of $(h1)$ \cite{DesAM:2017} has
$n\bar{p}_n(1-\bar{p}_n)\to +\infty$. Note that the difference only matters if $p_n\to 1$ which is the complete graph limit. 
Still using the same mapping to matrix factorization, \cite{LeiM:2018} treats the  
asymmetric case, however in a limit where $n\to +\infty$ first and $d_n\to +\infty$ after (in fact this anlaysis can accomodate any growth slower than $d_n\approx n^{1/2}$)
but it is unclear whether this is possible for denser regimes.
Our analysis covers this gap and the whole spectum of growth for $d_n$ up to linear growth is allowed. Besides, we propose a self-contained and direct method using the 
adaptive interpolation method \cite{BarM:2018}. A technical limitation of interpolation methods has often been the need to 
use Gaussian integration by parts. We by-pass this limitation using an (approximate) integration by parts formula 
for the edge binary variables $G_{ij}\in \{0,1\}$.

Before we formulate the adaptive interpolation let us set up more explicitly the quantities that we compute.
The distribution of $G$ given the hidden partition $X$ is the inhomogeneous Erdoes-R\'enyi graph measure: 
$$
\mathbb{P}(\bm{G}\vert \bm{X}) = \prod_{i<j} (\bar p_n + \Delta_n X_i X_j)^{G_{ij}} (1 - \bar p_n - \Delta_n X_i X_j)^{1 - G_{ij}}\, .
$$
Using this measure and Bayes rule, we find the posterior distribution of the SBM
\begin{align*}
\mathbb{P}(\bm{X} =\bm{x}| \bm{G}) &=\mathbb{P}(\bm{x} | \bm{G})= \frac{\mathbb{P}(\bm{G}| \bm{x}) \mathbb{P}(\bm{x})}{\mathbb{P}(\bm{G})} \propto \mathbb{P}(\bm{G}| \bm{x}) \mathbb{P}(\bm{x}) \\
	& =  	
	\exp \Big\{\sum_{i<j} \Big ( G_{ij} \ln (\bar{p}_n+\Delta_n x_i x_j ) + (1-G_{ij}) \ln (1-\bar{p}_n-\Delta_n x_i x_j) \Big )\Big\}\prod_{i=1}^{n} \mathbb{P}_{r}(x_i)  \\
	& =  \exp  \Big\{\sum_{i<j} \Big ( G_{ij} \ln (1+\frac{\Delta_n}{\bar{p}_n} x_i x_j ) + (1-G_{ij}) \ln (1-\frac{\Delta_n}{1-\bar{p}_n} x_i x_j) \Big ) + D_n(\bar{p}_n, \bm{G}) \Big \}  \prod_{i=1}^{n} \mathbb{P}_{r}(x_i)
\end{align*}
where $D_n(\bar{p}_n, \bm{G}) \equiv \sum_{i<j} G_{ij} \ln \bar{p}_n + (1-G_{ij}) \ln (1-\bar{p}_n)$.
Therefore, the posterior distribution becomes
\begin{align*}
\mathbb{P}(\bm{x} | \bm{G}) 
	& = \frac{1}{{\cal Z}({\bm{G}})} e^{-\sH_{\rm SBM}(\bm{x};\bm{G})} \prod_{i=1}^n\mathbb{P}_r(x_i)\,,\nonumber \\
\sH_{\rm SBM}(\bm{x};\bm{G}) 
	& \equiv - \sum_{i<j} \Big\{ G_{ij} \ln (1+x_i x_j\frac{\Delta_n}{\bar{p}_n} ) + (1-G_{ij}) \ln (1-x_i x_j \frac{\Delta_n}{1-\bar{p}_n} ) \Big \}\,.
\end{align*}
We use the statistical mechanics terminology and therefore call this posterior distribution the Gibbs distribution. The 
normalizing factor $${\cal Z}({\bm{G}})\equiv\sum_{\bm{x}\in\sX^n}  e^{-\sH_{\rm SBM}(\bm{x};\bm{G})} \prod_{i=1}^n\mathbb{P}_r(x_i)$$ is the partition function, and $\sH_{\rm SBM}$ is the Hamiltonian.
A straightforward computation, using the scaling regime \ref{h1} and \ref{h2}, gives the following formula (see the proof in Appendix~\ref{app:MI}):

\begin{proposition}[Linking the mutal information and log-partition function]
For the SBM under concern we have
\begin{align}
\frac{1}{n} I(\bm{X}; \bm{G}) 
	& = - \frac{1}{n} \mathbb{E}_{\bm{X}} \mathbb{E}_{\bm{G} | \bm{X}} \ln {\cal Z}({\bm{G}}) + \frac{\lambda_n}{4} + o_n(1) \label{eq:MI}
\end{align}
where $\lim_{n \rightarrow \infty} o_n(1) = 0$.
\label{thm:MI}
\end{proposition}

Thus the problem boils down to compute minus the expected log-partition function, or expected free energy, in the limit $n\to +\infty$. This will be achieved via an interpolation
towards the log-partition function of $n$ independent scalar Gaussian channels where the observations about the hidden labels are of the form
\begin{align}
Y_i = \sqrt{q}\,{X}_i + Z_i\,, \qquad 1\le i \le n\,,
\label{eq:Y}
\end{align}
with $Z_i \sim \mathcal{N}(0,1)$ i.i.d. Gaussian random variables and $q >0$ the signal-to-noise ratio (SNR). 
An important feature of our technique is the freedom to adapt a suitable interpolation path to the problem at hand. This is explained in the next section.

\section{Adaptive path interpolation}

We design an {\it interpolating model} parametrized by $t\in [0,1]$ and $\epsilon\geq 0$ s.t. at $t=\epsilon=0$ we recover the original SBM, while at $t=1$ we have a decoupled channel similar to \eqref{eq:Y}. For $t\in (0,1)$ the model is a mixture of the SBM with parameters $(\bar{p}_n, \sqrt{1-t}\,\Delta_n)$ and the extra {\it decoupled} Gaussian observations
\eqref{eq:Y} with SNR replaced by $$q\to R(t,\epsilon) \equiv \epsilon + \int_0^t ds\,q(s,\epsilon)$$ with $q(s,\epsilon) \geq 0$. 
The transition kernels for the channels $\bm{X} \rightarrow \bm{G}$ and $\bm{X} \rightarrow \bm{Y}$ at time $t \in [0,1]$ are
\begin{align}
\mathbb{P}_{t}(\bm{G} | \bm{X}) 
	& = \prod_{i<j} ( \bar{p}_n + \sqrt{1-t} \Delta_n X_i X_j )^{G_{ij}} ( 1 - \bar{p} - \sqrt{1-t} \Delta_n X_i X_j )^{1-G_{ij}} \nonumber \\
	& = \exp  \sum_{i<j}\Big( G_{ij} \ln (\bar{p} + \sqrt{1-t} \Delta_n X_i X_j) + (1-G_{ij}) \ln ( 1 - \bar{p}_n - \sqrt{1-t} \Delta_n X_i X_j ) \Big )\,, \label{eq:P-t-SBM} \\
\mathbb{P}_t(\bm{Y}|\bm{X}) 
	& = \frac{1}{(2\pi)^{n/2}} \exp \Big ( - \frac{1}{2} \sum_{i=1}^{n} (Y_i - \sqrt{R(t,\epsilon)} X_i )^2 \Big )\,. \label{eq:Y-density}
\end{align}
We constrain 
$\epsilon\in [s_n, 2s_n]$ where $s_n\to 0_+$ as $n\to +\infty$ at an appropriate rate to be fixed later on. The interpolating Hamiltonian is then defined to be
\begin{align*}
\sH_{t,\epsilon}(\bm{x};\bm{G}, \bm{Y}) \equiv \sH_{\mathrm{SBM};t}(\bm{x};\bm{G}) + \sH_{\mathrm{dec};t, \epsilon}(\bm{x};\bm{Y}) 
\end{align*}
where
\begin{align}
\sH_{\mathrm{SBM};t}(\bm{x};\bm{G}) & \equiv - \sum_{i<j} \Big( G_{ij} \ln(1+x_i x_j\sqrt{1-t}\frac{\Delta_n}{\bar{p}_n}) + (1-G_{ij}) \ln (1 -  x_i x_j\sqrt{1-t}\frac{\Delta_n}{1-\bar{p}_n})\Big)\,, \label{eq:H-SBM-t} \\ 
\sH_{\mathrm{dec};t,\epsilon}(\bm{x};\bm{Y}(\bm{X},\bm{Z})) & \equiv -\sum_{i=1}^{n}\Big(\sqrt{R(t,\epsilon)}Y_ix_i-R(t,\epsilon) \frac{x_i^2}{2}\Big) \nonumber \\
&=-\sum_{i=1}^{n}\Big(R(t,\epsilon)X_ix_i +\sqrt{R(t,\epsilon)}Z_ix_i-R(t,\epsilon) \frac{x_i^2}{2}\Big)\,. \label{eq:H-Y-t}
\end{align}
The posterior distribution expressed with the Hamiltonian $\sH_{t,\epsilon}(\bm{x};\bm{G}, \bm{Y})$ then reads
\begin{align*}
\mathbb{P}_{t}(\bm{x} | \bm{G}, \bm{Y}) = \frac{ \prod_{i=1}^{n} \mathbb{P}_{r}(x_i) \exp(-\sH_{t,\epsilon}(\bm{x};\bm{G},\bm{Y}))}{\sum_{\bm{x} \in \sX^n} \prod_{i=1}^{n} \mathbb{P}_{r}(x_i) \exp(-\sH_{t,\epsilon}(\bm{x};\bm{G},\bm{Y}))}\,.
\end{align*}
Therefore the {\it Gibbs-bracket} (i.e., the expectation operator w.r.t. the posterior distribution) for the interpolating model is 
\begin{align*}
\< A \>_{t,\epsilon} \equiv  \sum_{{\bm{x}} \in \sX^n } A({\bm x})\mathbb{P}_{t}(\bm{x} | \bm{G}, \bm{Y})=  \frac{1}{{\cal Z}_{t,\epsilon}({\bm G},{\bm Y})}\sum_{{\bm{x}} \in \sX^n } A({\bm x}) e^{-\sH_{t,\epsilon}({\bm{x}};{\bm{G}},{\bm{Y}})} \prod_{i=1}^n\mathbb{P}_r(x_i)
\end{align*}
with the partition function ${\cal Z}_{t,\epsilon}({\bm G},{\bm Y})\equiv \sum_{{\bm{x}} \in \sX^n }  e^{-\sH_{t,\epsilon}({\bm{x}};{\bm{G}},{\bm{Y}})}\prod_{i=1}^n\mathbb{P}_r(x_i)$. The reader should keep in mind that Gibbs-brackets are therefore functions of the {\it quenched} random variables $(\bm{Y}(\bm{X},\bm{Z}),\bm{G}(\bm{X}))$. The free energy for a given graph $\bm{G}=\bm{G}(\bm{X})$ (that depends on the ground truth partition) and decoupled observation $\bm{Y}(\bm{X},\bm{Z})$ is
\begin{align}
F_{t,\epsilon}(\bm{G},{\bm{Y}})=F_{t,\epsilon} & \equiv -\frac{1}{n} \ln  {\cal Z}_{t,\epsilon}({\bm G},{\bm Y})\,,
\label{free-en}
\end{align}
and its expectation 
\begin{align}
f_{t,\epsilon} \equiv \mathbb{E}_{{\bm{X}}} \mathbb{E}_{\bm{G} | {\bm{X}}} \mathbb{E}_{\bm{Y} | {\bm{X}}} F_{t,\epsilon} = \mathbb{E}_{{\bm{X}}} \mathbb{E}_{\bm{G} | {\bm{X}}} \mathbb{E}_{\bm{Z}}  F_{t,\epsilon}	\,.
\end{align}

By construction,
\begin{align*}
f_{t=0,\epsilon}
	& = - \frac{1}{n} \mathbb{E}_{\bm{X}} \mathbb{E}_{\bm{G} | \bm{X}} \mathbb{E}_{\bm{Z}} \ln \Big ( \sum_{\bm{x} \in \sX^n}  \exp \Big \{ \sum_{i<j} \Big ( G_{ij} \ln (1+\frac{\Delta_n}{\bar{p}_n} x_i x_j) + (1-G_{ij}) \ln (1-\frac{\Delta_n}{1-\bar{p}_n} x_i x_j) \Big ) \nonumber \\
	& \qquad+ \sum_{i=1}^{n} (\sqrt{\epsilon} Z_i x_i + \epsilon X_i x_i - \frac{\epsilon}{2} x_i^2 ) \Big \}\prod_{i=1}^{n} \mathbb{P}_{r}(x_i)\Big)\,, \label{eq:f_0e} \\
f_{t=1,\epsilon}
	& = - \frac{1}{n} \mathbb{E}_{\bm{Z}} \ln \Big ( \sum_{\bm{x} \in \sX^n}  \exp \Big \{ \sum_{i=1}^{n} (\sqrt{R(1,\epsilon)} Z_i x_i + R(1,\epsilon) X_i x_i - \frac{R(1,\epsilon)}{2} x_i^2 ) \Big \}\prod_{i=1}^{n} \mathbb{P}_{r}(x_i)\Big) \\
	& = \Psi(R(1,\epsilon), \lambda_n, r) - \frac{\lambda_n}{4} - \frac{R(1,\epsilon)^2}{4\lambda_n}\,.
\end{align*}
In particular, when $t=\epsilon = 0$ we have
\begin{align*}
f_{0,0}
	& = \frac{1}{n} I(\bm{X}; \bm{G}) - \frac{\lambda_n}{4} + o_n(1)\,.
\end{align*}
Therefore 
\begin{align}
\frac{1}{n} I(\bm{X}; \bm{G})
	& = f_{0,0} + \frac{\lambda_n}{4} + o_n(1) \nonumber \\
	& = \Psi(R(1,\epsilon), \lambda_n, r) - \frac{R(1,\epsilon)^2}{4\lambda_n} - f_{1,\epsilon} + f_{0,0} + o_n(1) \\
	& = \Psi(R(1,\epsilon), \lambda_n, r) - \frac{R(1,\epsilon)^2}{4\lambda_n} - \int_0^1 dt \frac{df_{t,\epsilon}}{dt} + (f_{0,0} - f_{0,\epsilon}) + o_n(1) \label{eq:sum-rule:0}
\end{align}
where $o_n(1)$ collects all contributions that tend to zero uniformly in $\epsilon$ when $n \rightarrow \infty$.
Eventually, we reach the following fundamental sum rule (see section~\ref{sec:sum-rule} for the derivation):
\begin{align}
\frac{1}{n} I(\bm{X}; \bm{G}) 
	& = \Psi(R(1,\epsilon), \lambda_n, r) + \sR_1 - \frac{1}{4\lambda_n} \int_0^1 dt\, \sR_2(t) - \sR_{3} \label{eq:sum-rule:1}
\end{align}
where
\begin{align*}
\sR_1 & \equiv \frac{1}{4\lambda_n} \Big ( \int_0^1 q(t,\epsilon)^2 dt - \Big(\int_0^1 q(t,\epsilon) dt\Big)^2 \Big ) \geq 0\,, \\
\sR_2(t) & \equiv \mathbb{E} \< ( \lambda_n Q - q(t,\epsilon) )^2 \>_{t,\epsilon} \geq 0\,, \\
\sR_3 & \equiv \frac{\epsilon}{4\lambda_n} \Big(\epsilon + 2 \int_0^1 q(t,\epsilon) dt \Big) - \frac{1}{2} \int_{0}^{\epsilon} d\epsilon'\, \mathbb{E} \< Q \>_{0,\epsilon'} + o_n(1)\,,
\end{align*}
and the {\it overlap} is $$Q(\bm{X},\bm{x}) = Q \equiv \frac{1}{n} \sum_{i=1}^{n} X_i x_i\,.$$

Two generic tools that we will widely use in our proof are the following:
\begin{itemize}
\item {\bf The Nishimori identity}: Let $(X,Y)$ be a couple of random variables with joint distribution $P(X, Y)$ and conditional distribution 
$P(\cdot | Y)$. Let $k \geq 1$ and let $x^{(1)}, \dots, x^{(k)}$ be i.i.d.\ copies from the conditional distribution. Let us denote $\langle - \rangle$ the expectation 
 w.r.t. the product distribution $P(\cdot | Y)^{\otimes \infty}$ over copies and $\mathbb{E}$ the expectation w.r.t. the joint distribution. 
 Then, for all continuous bounded functions $g$ we have
\begin{align*}
\mathbb{E} \langle g(Y,x^{(1)}, \dots, x^{(k)})\rangle
=
\mathbb{E}\langle g(Y, X, x^{(2)}, \dots, x^{(k)}) \rangle\,. 
\end{align*}
The expectation $\mathbb{E}$ is over $(X, Y)$.	
\begin{proof}
This is a simple consequence of Bayes formula.
It is equivalent to sample the couple $(X,Y)$ according to its joint distribution or to sample first $Y$ according to its marginal distribution and then to sample $x$ 
conditionally on $Y$ from the conditional distribution. Thus the two $(k+1)$-tuples $(Y,x^{(1)}, \dots,x^{(k)})$ and $(Y, X, x^{(2)},\dots,x^{(k)})$ have the same law.	
\end{proof}

In the present case $(X, Y) \to (\bm{X}, \bm{G}, \bm{Y})$ with joint 
law $\mathbb{P}_t(\bm{X}\vert \bm{G}, \bm{Y}) \prod_{i=1}^n\mathbb{P}_r(X_i)$. Let us take $k$ i.i.d. copies $\bm{x}^{(1)}, \dots, \bm{x}^{(k)}$ drawn 
from the posterior distribution $\mathbb{P}_{t}(\cdot|\bm{G}, \bm{Y})$. Then for any continuous bounded function $g$ 
\begin{align}
\mathbb{E}\< g( \bm{G}, \bm{Y}, \bm{x}^{(1)}, \dots, \bm{x}^{(k-1)}, \bm{X} ) \>_{t,\epsilon} = \mathbb{E} \< g( \bm{G}, \bm{Y}, \bm{x}^{(1)}, \dots, \bm{x}^{(k-1)}, \bm{x}^{k} ) \>_{t,\epsilon} \,.\label{eq:Nishimori}
\end{align}
where $\mathbb{E}$ is over $(\bm{G}, \bm{Y})$. More precisely 
$\mathbb{E} = \mathbb{E}_{\prod_{i=1}^n\mathbb{P}_r(X_i)}\mathbb{E}_{\mathbb{P}_t(\bm{G}\vert\bm{X})}\mathbb{E}_{\mathbb{P}_t(\bm{Y}\vert \bm{X}})$.
Note that, by a slight abuse of notation, we continue to use the Gibbs-bracket notation for expressions depending on multiple i.i.d. copies from the 
posterior, so that $\< - \rangle_{t,\epsilon}$ corresponds to the expectation w.r.t. the product measure $\mathbb{P}_{t}(\cdot|\bm{G}, \bm{Y})^{\otimes \infty}$.

\item {\bf Gaussian integration by parts}: Integration by parts implies that for any bounde differentiable function $g$ of $Z \sim \mathcal{N}(0,1)$ we have
\begin{align}
\mathbb{E}[Zg(Z)] = \mathbb{E}\,[g'(Z)]\,.
\label{eq:stein}
\end{align}
\end{itemize}

We are now ready to provide the proofs of the bounds on the mutual information.

\subsection{The upper bound: proof of Theorem \ref{thm:ub}}

Set $\epsilon = 0$ and $q(t,\epsilon)=q$ a non-negative constant. Then we have $\sR_1=0$, $\sR_3=o_n(1)$. Since $\sR_2\geq 0$, \eqref{eq:sum-rule:1} implies 
\begin{align*}
\frac{1}{n} I(\bm{X}; \bm{G}) 
	\leq \Psi(q, \lambda_n, r) + o_n(1)\,.
\end{align*}
Since $\Psi$ is continuous w.r.t its second argument $\limsup_{n\to +\infty} \frac{1}{n} I(\bm{X}; \bm{G}) \leq \Psi(q, \lambda, r)$. Optimizing over $q\in [0, \lambda]$
yields the bound (optimization over $q\in [0, +\infty)$ does not yield a sharper bound, see remark 2). 


\subsection{The lower bound: proof of Theorem \ref{thm:lb}}

The basic idea is to ``remove'' $\sR_2$ from \eqref{eq:sum-rule:1} by {\it adapting} $q(t,\epsilon)$. Then taking the limit $n \rightarrow \infty$ and $\epsilon \rightarrow 0_+$ will provide 
the desired bound since $\sR_1 \geq 0$ and $\sR_3\to 0$ will disappear. To implement this idea we first decompose $\sR_2$ into
\begin{align}
\sR_2 = (\lambda_n \mathbb{E}\<Q\>_{t,\epsilon} - q(t,\epsilon))^2 + \lambda_n^2 \mathbb{E} \< ( Q - \mathbb{E}\<Q\>_{t,\epsilon} )^2 \>_{t,\epsilon} \label{eq:decomp}
\end{align}
and address each part with the following two lemmas. The proof of Lemma~\ref{thm:overlap} can be found in section~\ref{sec:overlap}.
\begin{lemma}
For every $\epsilon \in [0,1]$ and $t\in [0,1]$ there exists a (unique) bounded solution $R_n^*(t,\epsilon)=\epsilon +\int_0^t ds\, q_n^*(s,\epsilon)$ to the first order differential equation
\begin{align}
\frac{dR}{dt}(t,\epsilon) =\lambda_n \mathbb{E} \< Q \>_{t,\epsilon} \quad \text{with} \quad R(0,\epsilon) = \epsilon\,.
\label{eq:m-star}
\end{align}
Furthermore
$$
q_n^*(t,\epsilon)=\lambda_n \mathbb{E} \< Q \>_{t,\epsilon}\in [0, \lambda_n]\,, \qquad \text{and}\qquad \frac{dR^*_n}{d\epsilon}(t,\epsilon) \geq 1\,.
$$
\label{thm:m-star}
\end{lemma}
\begin{proof}
Let $G_n(t,R(t,\epsilon)) \equiv \lambda_n \mathbb{E} \< Q \>_{t,\epsilon}$. Equation~\eqref{eq:m-star} is thus a first-order differential equation.
Also note that, letting $dG_n/dR$ be the derivative w.r.t. the second argument, 
\begin{align}
\frac{dG_n}{dR}(t,R(t,\epsilon))&=\frac{\lambda_n}{n} \sum_{i=1}^{n} \mathbb{E}\Big [ X_i \sum_{\bm{x}\in{\cal X}^n}x_i\mathbb{P}_r(\bm{x}) \frac{d}{dR} \frac{e^{-{\cal H}_{t,\epsilon}(\bm{x};\bm{G},\bm{Y})}}{{\cal Z}_{t,\epsilon}({\bm G},{\bm Y})}  \Big ]\\
&=\frac{\lambda_n}{n} \sum_{i=1}^{n} \mathbb{E}\Big [ X_i \sum_{\bm{x}\in{\cal X}^n}x_i\mathbb{P}_r(\bm{x}) \nonumber \\
&\qquad\times\Big(-\frac{e^{-{\cal H}_{t,\epsilon}(\bm{x};\bm{G},\bm{Y})}}{{\cal Z}_{t,\epsilon}({\bm G},{\bm Y})}\frac{d{\cal H}_{t,\epsilon}(\bm{x};{\bm G},{\bm Y})}{dR}- \frac{e^{-{\cal H}_{t,\epsilon}(\bm{x};\bm{G},\bm{Y})}}{{\cal Z}_{t,\epsilon}({\bm G},{\bm Y})}\frac{\frac{d}{dR}{\cal Z}_{t,\epsilon}({\bm G},{\bm Y})}{{\cal Z}_{t,\epsilon}({\bm G},{\bm Y})} \Big) \Big ]\nonumber\\
	& = \frac{\lambda_n}{n} \sum_{i,j=1}^{n} \mathbb{E}\Big [ X_i \Big\< x_i ( x_j X_j + \frac{x_j Z_j}{2 \sqrt{R(t,\epsilon)}} - \frac{x_j^2}{2} ) \Big\>_{t,\epsilon} - X_i \< x_i \>_{t,\epsilon} \Big\< x_j X_j + \frac{x_j Z_j}{2 \sqrt{R(t,\epsilon)}} - \frac{x_j^2}{2} \Big\>_{t,\epsilon} \Big ] \nonumber\\
	& = \frac{\lambda_n}{2n} \sum_{i,j=1}^{n} \mathbb{E} \Big [ 2 X_i X_j \< x_i x_j \>_{t,\epsilon} - X_i \< x_i x_j \>_{t,\epsilon} \< x_j \>_{t,\epsilon} \nonumber \\
	& \qquad - 2 X_i X_j \< x_i \>_{t,\epsilon} \< x_j \>_{t,\epsilon} + 2 X_i \<x_i \>_{t,\epsilon} \< x_j \>_{t,\epsilon}^2 - X_i \< x_i x_j \>_{t,\epsilon} \< x_j \>_{t,\epsilon}\Big] \label{last}
\end{align}
To get the last identity, we used Gaussian integration by parts, which reads when applied to Gibbs brackets,
\begin{align*}
\EE[Z_j\langle f\rangle_{t,\epsilon}] = \sqrt{R(t,\epsilon)} \EE[\langle fx_j\rangle_{t,\epsilon}-\langle f \rangle_{t,\epsilon} \langle x_j\rangle_{t,\epsilon}]\,.
\end{align*}
Indeed, one must be careful that in the definition of the Gibbs bracket both the Hamiltonian {\it and} partition function are 
functions of the quenched variable $\bm Z$, thus the appearance of two terms when we differentiate w.r.t $Z$. Now, using the Nishimori identity to 
replace the hidden partition $\bm X$ by a new independent sample from the posterior in \eqref{last} 
(which yields, e.g., $\EE[X_i X_j \< x_i x_j \>_{t,\epsilon}]=\EE[\< x_i x_j \>_{t,\epsilon}^2]$ 
or $\EE[X_i \< x_i x_j \>_{t,\epsilon} \< x_j \>_{t,\epsilon}]=\EE[\langle x_i\rangle_{t,\epsilon} \< x_i x_j \>_{t,\epsilon} \< x_j \>_{t,\epsilon}]$) 
we reach
\begin{align}
\frac{dG_n}{dR}(t,R(t,\epsilon)) & = \frac{\lambda_n}{n} \sum_{i,j=1}^{n} \mathbb{E}[(\< x_i x_j \>_{t,\epsilon} - \< x_i \>_{t,\epsilon} \< x_j \>_{t,\epsilon})^2]\,.\label{derBound}
\end{align}
The function $G_n$ is bounded and takes values in $[0,\lambda_n]$. 
Indeed $\EE\langle Q \rangle_{t,\epsilon} = \EE[X_1\langle x_1 \rangle_{t,\epsilon}] = \EE[\langle x_1 \rangle_{t,\epsilon}^2]$ 
by the Nishimori identity, thus $\EE\langle Q \rangle_{t,\epsilon}\le \EE\langle x_1^2 \rangle_{t,\epsilon}=\EE[X_1^2]$ again by the Nishimori identity, and finally $\EE[X_1^2]=1$. 
In addition of being bounded, $G_n$ is differentiable w.r.t. its second argument, with bounded derivative 
as seen from \eqref{derBound}. The Cauchy-Lipschitz theorem then implies that \eqref{eq:m-star} admits a unique global solution over $t\in [0,1]$. 
Finally Liouville's formula (see Appendix~\ref{app:Liouville}) gives
%
\begin{align}
\frac{dR_n^*}{d\epsilon}(t,\epsilon) = \exp  \int_0^t dt' \frac{dG_n}{dR}(t',R_n^*(t',\epsilon)) \,. \label{eq:Liouville}
\end{align}
The non-negativity of $dG_n/dR$ then implies $dR_n^* /d\epsilon \geq 1$.
\end{proof}

We now state a crucial concentration result for the overlap. Its validity is a consequence of the fact that the problem is analyzed in the so-called Bayesian optimal setting. This means that all hyper-parameters in the problem, namely $(\mathbb{P}_r, r,\bar{p}_n,\Delta_n)$, are assumed to be known, so that the posterior of the model can be written exactly. It implies the validity of the Nishimori identity which in turn allows to prove the following result (see section \ref{sec:overlap}):

\begin{lemma}[Overlap concentration]\label{thm:overlap}
Let $R$ be the solution $R_n^*$ in Lemma~\ref{thm:m-star}. Then for any bounded positive sequence $s_n$ there exists a sequence $C_n(r, \lambda_n) >0$ converging to a constant and such that
\begin{align*}
\frac{1}{s_n} \int_{s_n}^{2s_n} d\epsilon\,\mathbb{E}\< ( Q - \mathbb{E} \< Q \>_{t,\epsilon} )^2 \>_{t,\epsilon}   \leq \frac{C_n(r,\lambda_n)}{(s_n^4 n)^{1/3}}\,.
\end{align*}
\end{lemma}


Now we average \eqref{eq:sum-rule:1} over a small interval $\epsilon \in [s_n, 2s_n]$ (note that $I(\bm{X}; \bm{G})$ is 
independent of $\epsilon$) and set $R$ to the solution $R_n^*$ of \eqref{eq:m-star} in Lemma \ref{thm:m-star}; 
therefore $q_n^*(t,\epsilon)=\lambda_n \mathbb{E} \< Q \>_{t,\epsilon}$. This choice cancels the first term of $\sR_2$ in the 
decomposition \eqref{eq:decomp}. The second term in \eqref{eq:decomp} is then upper bounded using Lemma \ref{thm:overlap}. 
Finally $\sR_1\ge 0$. Combining all these observations we obtain
\begin{align}
 \frac{1}{n} I(\bm{X}; \bm{G})\ge \frac{1}{s_n}\int_{s_n}^{2s_n} d\epsilon [\Psi(R_n^*(1,\epsilon), \lambda_n, r)-\sR_3]- \frac{C_n(r,\lambda_n)\lambda_n}{4(s_n^4 n)^{1/3}}  \label{sumRule:end}
\end{align}
where we used Fubini's theorem to switch the $t$ and $\epsilon$ integrals when using 
Lemma \ref{thm:overlap}. Using $q_n^*\in [0, \lambda_n]$ and $\epsilon\in[s_n,2s_n]$, we see that $\sR_3$ is bounded uniformly in $\epsilon$: 
\begin{align*}
|\sR_3 |
	\leq \frac{2s_n}{4\lambda_n} (2s_n + 2 \lambda_n) + o_n(1)
	= \frac{s_n^2}{\lambda_n} + s_n + o_n(1)\,.	
\end{align*}
Therefore the average of $\sR_3$ over $\epsilon$ has the same upper bound. Now, since
$$
\frac{d}{d\lambda} \Psi(R_n^*(1,\epsilon), \lambda, r) = \frac{1}{4} - \frac{R_n^*(1,\epsilon)^2}{4\lambda} 
$$
and 
$R_n^*(1,\epsilon)\in[s_n,2s_n+ \lambda_n]$ we have
$-\frac{1}{4} \leq\frac{d}{d\lambda} \Psi(R_n^*(1,\epsilon), \lambda)\leq \frac{1}{4}$ (we use $n$ large enough for the l.h.s inequality).
Thus by remark 2 and the mean value theorem
\begin{align*}
 \frac{1}{s_n}\int_{s_n}^{2s_n} d\epsilon \,\Psi(R_n^*(1,\epsilon), \lambda_n, r)
  & = 
  \frac{1}{s_n}\int_{s_n}^{2s_n} d\epsilon \,\Psi(R_n^*(1,\epsilon), \lambda, r)
  + 
  \frac{1}{s_n}\int_{s_n}^{2s_n} d\epsilon \,( \Psi(R_n^*(1,\epsilon), \lambda_n, r) - \Psi(R_n^*(1,\epsilon), \lambda, r))
  \nonumber \\ &
  \geq 
  \min_{q\in [0, \lambda]} \Psi(q, \lambda, r)  - \frac{1}{4} \vert \lambda_n - \lambda\vert
\end{align*}
These remarks imply a relaxation of \eqref{sumRule:end}:
\begin{align}
\frac{1}{n} I(\bm{X}, \bm{G}) 
	\geq \min_{q \in [0,\lambda]} \Psi(q, \lambda_n, r) - \frac{1}{4}\vert \lambda_n - \lambda\vert -\frac{C_n(r,\lambda_n)\lambda_n}{4(s_n^4 n)^{1/3}} - \frac{s_n^2}{\lambda_n} - s_n  - o_n(1)\,. \label{eq:sum-rule-lb}
\end{align}
Finally,
setting $s_n = n^{-\theta}$ with $\theta \in (0, 1/4)$ ensures the extra terms on the r.h.s. 
of \eqref{sumRule:end} vanish as $n\to+\infty$. Then taking the $\liminf_{n\to +\infty}$  and using $\lambda_n\to\lambda$ we finally reach the desired bound. 

%
\section{The fundamental sum rule: proof of \eqref{eq:sum-rule:1}}
\label{sec:sum-rule}

In this section we use the notation $F_{t,\epsilon}$ for \eqref{free-en} without explicitly indicating the dependence in its arguments. When $G_{ij}$ is set to zero for a specific pair $(i,j)$ all other 
$G_{k,l}$, $(k,l)\neq (i,j)$ being fixed
we write $F_{t,\epsilon}(G_{ij}=0)$. Expectation with respect to the set of all $G_{k,l}$, $(k,l)\neq (i,j)$ is denoted by $\mathbb{E}_{\sim G_{ij}}$. 

The derivative of the averaged free energy can be decomposed into three terms: 
\begin{align}
\frac{df_{t,\epsilon}}{dt} = D_1 + D_2 + D_3 \label{eq:chain-rule}
\end{align}
where
\begin{align*}
D_1 & \equiv \mathbb{E}_{{\bm{X}}} \mathbb{E}_{\bm{Y} | {\bm{X}}}  \sum_{\bm{G}} F_{t,\epsilon} \frac{d}{dt} \mathbb{P}_{t}(\bm{G} | {\bm{X}})\,, \nonumber \\
D_2 & \equiv \mathbb{E}_{{\bm{X}}} \mathbb{E}_{\bm{G} | {\bm{X}}}  \int d\bm{Y} F_{t,\epsilon} \frac{d}{dt} \mathbb{P}_{t}(\bm{Y} | {\bm{X}})\,, \\	
D_3 & \equiv \frac{1}{n} \mathbb{E} \Big\< \frac{d}{dt} \sH_{\mathrm{dec}; t, \epsilon}\Big\>_{t,\epsilon}+\frac{1}{n} \mathbb{E} \Big\< \frac{d}{dt} \sH_{\mathrm{SBM}; t}\Big\>_{t,\epsilon}\,.
\end{align*}

\subsection{Term $D_1$.}

\begin{lemma} \label{4.1}
We have $D_1
	 = \frac{\lambda_n}{4} \mathbb{E}\<Q^2\>_{t,\epsilon} + \mathcal{O}(\frac{1}{n}) + \mathcal{O}\big ( \frac{\lambda_n^{3/2}}{\sqrt{n \bar{p}_n(1-\bar{p}_n)^{3}}} \big )$.
\end{lemma}
\begin{proof}
Note that by \eqref{eq:P-t-SBM} we have
\begin{align*}
\frac{d}{dt} \mathbb{P}_{t}(\bm{G} | \bm{X})
	& = \mathbb{P}_{t}(\bm{G} | \bm{X}) \sum_{i<j} \frac{1}{2} \frac{\Delta_n}{\sqrt{1-t}} X_i X_j \Big ( - \frac{G_{ij}}{ \bar{p}_n + \sqrt{1-t} \Delta_n X_i X_j } + \frac{1-G_{ij}}{1 - \bar{p}_n - \sqrt{1-t} \Delta_n X_i X_j } \Big ) \,.
\end{align*}
This gives
\begin{align}
	D_1 & = \frac{\Delta_n}{2\sqrt{1-t}} \sum_{i<j} \mathbb{E}_{\bm{X}} \mathbb{E}_{\bm{Y} | \bm{X}} \mathbb{E}_{\bm{G} | \bm{X}} 
	\Big [  X_i X_j \bigg( \frac{(1-G_{ij}) F_{t,\epsilon}}{1 - \bar{p}_n- \sqrt{1-t} \Delta_n X_i X_j }-\frac{G_{ij} F_{t,\epsilon}}{ \bar{p}_n + \sqrt{1-t} \Delta_n X_i X_j }  \bigg)\Big ] 
	\nonumber \\
	& = \frac{\Delta_n}{2\sqrt{1-t}} ( D_1^{(a)} + D_1^{(b)} ) \label{eq:comp4:0}
\end{align}
with the definitions
\begin{align*}
	D_1^{(a)} & \equiv \sum_{i<j} \mathbb{E}_{\sim G_{ij}} 
	\bigg[ {X}_i {X}_j \frac{ \mathbb{E}_{G_{ij} | {\bm{X}}} F_{t,\epsilon} 
	- \mathbb{E}_{G_{ij} |{\bm{X}}} [G_{ij} F_{t,\epsilon}] }{1 - \mathbb{E}_{G_{ij} | {X_i, X_j}} G_{ij} }  \bigg] \,,\\
	D_1^{(b)} & \equiv - \sum_{i<j} \mathbb{E}_{\sim G_{ij}} 
	\bigg[{X}_i{X}_j  \frac{ \mathbb{E}_{G_{ij} | {\bm{X}}} 
	[G_{ij} F_{t,\epsilon}] }{ \mathbb{E}_{G_{ij} | {X_i, X_j}} G_{ij} } \bigg]\,,
\end{align*}
where $\mathbb{E}_{\sim G_{ij}} \equiv \mathbb{E}_{\bm{X}} \mathbb{E}_{\bm{Y} | \bm{X}} \mathbb{E}_{\bm{G} \setminus G_{ij}  | \bm{X}}$, and recalling $$\mathbb{E}_{G_{ij} | X_i, X_j} G_{ij} = \bar{p}_n + \sqrt{1-t} \Delta_n X_i X_j\,.$$

Both $D_1^{(a)}$ and $D_1^{(b)}$ involve the term $\mathbb{E}_{G_{ij} |{\bm{X}}} [G_{ij} F_{t,\epsilon}]$. In Section~\ref{sec:approxint} we derive 
an \emph{approximate integration by parts formula} that, when applied in the present case, yields

\begin{lemma}\label{lemm:intparts}
 Fix $i,j \in \{1,\cdots,n\}^2$ and recall that $G_{ij}\in\{0,1\}$ with conditional mean $\mathbb{E}_{G_{ij} | X_i, X_j}[G_{ij}] = \bar p_n + \sqrt{1-t} \Delta_n X_iX_j$. 
 Let $F_{t,\epsilon}^{(1)}(G_{ij})$ be the first partial derivative of $F_{t,\epsilon}$ with respect to $G_{ij}$. We have the approximate integration by parts formula
 \begin{align}
\mathbb{E}_{G_{ij} | X_i, X_j}[G_{ij} F_{t,\epsilon}(G_{ij})] 
	= & \mathbb{E}_{G_{ij} | X_i, X_j}[F_{t,\epsilon}^{(1)}(G_{ij})] \mathbb{E}_{G_{ij} | X_i, X_j}[G_{ij}] + F_{t,\epsilon}(G_{ij}=0) \mathbb{E}_{G_{ij} | X_i, X_j}[G_{ij}] 
	\nonumber \\ &
	+ 
	\mathcal{O} \Big ( \frac{\sqrt{1-t} \lambda_n}{n^2 (1-\bar{p}_n)} \Big )\,. \label{eq:talagrand}
\end{align}
 where
 $$
 F_{t,\epsilon}^{(1)}(G_{ij}) = - \frac{1}{n}\frac{\Delta_n}{\bar{p}_n (1-\bar{p}_n)}\sqrt{1-t} \langle x_i x_j\rangle_{t,\epsilon} + \mathcal{O}\Big(\frac{1}{n}\Big(\frac{\Delta_n}{\bar{p}_n (1-\bar{p}_n)}\Big)^2 (1-t)\Big)
 $$
 and $F_{t,\epsilon}(G_{ij}=0)$ is the evaluation of $F_{t,\epsilon}$ at $G_{ij}=0$ all other variables $G_{kl}$, $(k,l)\neq (i,j)$ being fixed.
\end{lemma}

The approximate integration by part formula \eqref{eq:talagrand} implies that the term $D_1^{(b)}$ of \eqref{eq:comp4:0} can be written as (recall $\bar{p}_n(1-\bar{p}_n) \gg \Delta_n$)
\begin{align}
	& \frac{\Delta_n}{2\sqrt{1-t}} D_1^{(b)} \nonumber \\
	& \quad = -\frac{\Delta_n}{2\sqrt{1-t}} \sum_{i<j} \mathbb{E}_{\sim G_{ij}} \Big [ X_i X_j \big ( F_{t,\epsilon}(G_{ij}= 0) - \frac{\sqrt{1-t} \Delta_n}{n\bar{p}_n(1-\bar{p}_n)} \mathbb{E}_{G_{ij} | X_i, X_j}\< x_i x_j \>_{t,\epsilon} \big ) \Big ] + \mathcal{O}\Big ( \frac{\lambda_n \Delta_n}{\bar{p}_n(1-\bar{p}_n)} \Big ) \nonumber \\
	& \quad = \frac{\Delta_n^2}{2 n \bar{p}_n(1-\bar{p}_n)} \sum_{i<j} \mathbb{E} [ X_i X_j \< x_i x_j \>_{t,\epsilon} ] - \frac{\Delta_n}{2\sqrt{1-t}} \sum_{i<j} \mathbb{E}_{\sim G_{ij}} [ X_i X_j F_{t,\epsilon}(G_{ij}= 0)] + \mathcal{O}\Big ( \frac{\lambda_n \Delta_n}{\bar{p}_n(1-\bar{p}_n)} \Big )\,. \label{eq:comp4:1a}
\end{align}
Applying again the approximate integration by parts formula \eqref{eq:talagrand} the term $D_1^{(a)}$ of \eqref{eq:comp4:0} can be written as (recall $(1-\bar{p}_n)^2 \gg \Delta_n$)
\begin{align}
	&\frac{\Delta_n}{2\sqrt{1-t}} D_1^{(a)} \nonumber \\
	&\quad = - \frac{\Delta_n}{2\sqrt{1-t}} \sum_{i<j} \mathbb{E}_{\sim G_{ij}} \Big [ X_i X_j \frac{\mathbb{E}_{G_{ij} | {X_i, X_j}} G_{ij}}{1 - \mathbb{E}_{G_{ij} | {X_i, X_j}} G_{ij}} \big ( F_{t,\epsilon}(G_{ij}=0) - \frac{\sqrt{1-t}\Delta_n}{n\bar{p}_n(1-\bar{p}_n)} \mathbb{E}_{G_{ij} | X_i, X_j}\< x_i x_j \>_{t,\epsilon} \big ) \Big ] \nonumber \\
	& \qquad + \frac{\Delta_n}{2\sqrt{1-t}} \sum_{i<j} \mathbb{E}_{\sim G_{ij}} \Big [ X_i X_j \frac{ \mathbb{E}_{G_{ij} | \bm{X}} F_{t,\epsilon} }{1 - \mathbb{E}_{G_{ij} | {X_i, X_j}} G_{ij} } \Big ] + \mathcal{O} \Big ( \frac{\lambda_n \Delta_n}{(1-\bar{p}_n)^2} \Big ) \nonumber \\
	&\quad = E_1 + E_2 + \frac{\Delta_n}{2\sqrt{1-t}} \sum_{i<j} \mathbb{E}_{\sim G_{ij}}[X_i X_j F_{t,\epsilon}(G_{ij}=0)] + \mathcal{O} \Big ( \frac{\lambda_n \Delta_n}{(1-\bar{p}_n)^2} \Big ) \, \label{eq:comp4:1b:0}
\end{align}
where we define
\begin{align*}
& E_1 \equiv \frac{\Delta_n}{2\sqrt{1-t}} \sum_{i<j} \mathbb{E}_{\sim G_{ij}} \Big [ X_i X_j \frac{\mathbb{E}_{G_{ij} | X_i X_j}F_{t,\epsilon} - F_{t,\epsilon}(G_{ij}=0) }{1 - \mathbb{E}_{G_{ij} | {X_i, X_j}} G_{ij} } \Big ], \\
& E_2 \equiv \frac{\Delta_n^2}{2n \bar{p}_n(1-\bar{p}_n)} \sum_{i<j} \mathbb{E} \Big [ \frac{\mathbb{E}_{G_{ij} | {X_i, X_j}} G_{ij}}{1 - \mathbb{E}_{G_{ij} | {X_i, X_j}} G_{ij}} X_i X_j \< x_i x_j \>_{t,\epsilon}  \Big ].
\end{align*}
We show in Appendix~\ref{app:D1} that in \eqref{eq:comp4:1b:0} the terms $E_1$ and $E_2$ approximately cancel so that
\begin{align}
\frac{\Delta_n}{2\sqrt{1-t}} D_1^{(a)} = \frac{\Delta_n}{2\sqrt{1-t}} \sum_{i<j} \mathbb{E}_{\sim G_{ij}}[X_i X_j F_{t,\epsilon}(G_{ij}=0)] + \mathcal{O}\Big ( \frac{\lambda_n \Delta_n}{(1-\bar{p}_n)^2} \Big ). \label{eq:comp4:1b}
\end{align}

Finally, substituting \eqref{eq:comp4:1a} and \eqref{eq:comp4:1b} into \eqref{eq:comp4:0} gives
\begin{align*}
 \mathbb{E}_{\bm{X}} \mathbb{E}_{\bm{Y} | \bm{X}} \sum_{\bm{G}} F_{t,\epsilon}\frac{d }{dt} \mathbb{P}_{t}(\bm{G} | \bm{X})
 	& = \frac{\Delta_n^2}{2 n \bar{p}_n (1-\bar{p}_n)} \sum_{i<j} \mathbb{E} [ X_i X_j \< x_i x_j \>_{t,\epsilon} ] + \mathcal{O}\Big ( \frac{\lambda_n \Delta_n}{\bar{p}_n(1-\bar{p}_n)} \Big ) +\mathcal{O}\Big ( \frac{\lambda_n \Delta_n}{(1-\bar{p}_n)^2} \Big ) \\
 	& = \frac{\lambda_n}{4} \mathbb{E}\<Q^2\>_{t,\epsilon} + \mathcal{O}\Big(\frac{1}{n}\Big) + \mathcal{O}\Big ( \frac{\lambda_n \Delta_n}{\bar{p}_n(1-\bar{p}_n)} \Big ) + \mathcal{O}\Big ( \frac{\lambda_n \Delta_n}{(1-\bar{p}_n)^2} \Big ) \\
 	& = \frac{\lambda_n}{4} \mathbb{E}\<Q^2\>_{t,\epsilon} + \mathcal{O}\Big(\frac{1}{n}\Big) + \mathcal{O}\Big ( \frac{\lambda_n^{3/2}}{\sqrt{n \bar{p}_n (1-\bar{p}_n)^{3}}} \Big ) \,,
\end{align*}
where, in the last two equalities, we used $\lambda_n = n\Delta_n^2/ (\bar{p}_n(1-\bar{p}_n))$ and
$Q = \frac1n \sum_{i=1}^{n} X_i x_i$. With \ref{h1} and \ref{h2}, all the error terms represented by the big-O notations tend to zero.
\end{proof}

\subsection{Term $D_2$.}

\begin{lemma}\label{4.2}
We have $D_2 = - \frac{1}{2} q(t,\epsilon) \mathbb{E}\<Q\>_{t,\epsilon}$. 
\end{lemma}
\begin{proof}
Recall \eqref{eq:Y-density}. Using Gaussian integration by parts \eqref{eq:stein} we obtain
\begin{align*}
D_2 \equiv \mathbb{E}_{\bm{X}} \mathbb{E}_{\bm{G} | \bm{X}}  \int d\bm{Y} F_{t,\epsilon}  \frac{d}{dt} \mathbb{P}_{t}(\bm{Y} | \bm{X}) 
	& = \sum_{i=1}^{n} \mathbb{E}_{\bm{X}} \mathbb{E}_{\bm{G} | \bm{X}} \mathbb{E}_{\bm{Y} | \bm{X}} \Big [ (Y_i - \sqrt{R(t,\epsilon)} X_i)  \frac{q(t,\epsilon) X_i}{2\sqrt{R(t,\epsilon)}}  F_{t,\epsilon} \Big ] \\
	& = \frac{q(t,\epsilon)}{2\sqrt{R(t,\epsilon)}} \sum_{i=1}^{n} \mathbb{E}_{\bm{X}} \mathbb{E}_{\bm{G} |  \bm{X}} \mathbb{E}_{\bm{Z}} \big [ Z_i X_i F_{t,\epsilon} \big ] \\
	& = - \frac{q(t,\epsilon)}{2n\sqrt{R(t,\epsilon)}} \sum_{i=1}^{n} \mathbb{E}_{\bm{X}} \mathbb{E}_{\bm{G} | \bm{X}} \mathbb{E}_{\bm{Z}} \big [ X_i \< \sqrt{R(t,\epsilon)} x_i \>_{t,\epsilon} \big ] \\
	& = - \frac{1}{2} q(t,\epsilon) \mathbb{E}\<Q\>_{t,\epsilon}\,,
\end{align*}
where we used that $\frac{dF_{t,\epsilon}}{dZ}=-\frac1n\< \sqrt{R(t,\epsilon)} x_i \>_{t,\epsilon}$, and then the definition of the overlap.
\end{proof}

\subsection{Term $D_3$.}

\begin{lemma}\label{4.3}
We have $D_3=0$.
\end{lemma}
\begin{proof}
Using the Nishimori identity \eqref{eq:Nishimori} we obtain
\begin{align*}
\mathbb{E} \Big\< \frac{d}{dt} \sH_{\mathrm{dec}; t, \epsilon}\Big\>_{t,\epsilon}
	& = -q(t,\epsilon) \sum_{i=1}^{n} \mathbb{E}_{\bm{X}} \mathbb{E}_{\bm{G} | \bm{X}} \mathbb{E}_{\bm{Y} | \bm{X}} \Big\< \frac{Y_i x_i}{2\sqrt{R(t,\epsilon)}} - \frac{x_i^2}{2} \Big\>_{t,\epsilon} \\
	& = -q(t,\epsilon) \sum_{i=1}^{n} \mathbb{E}_{\bm{X}} \mathbb{E}_{\bm{G} | \bm{X}} \mathbb{E}_{\bm{Y} | \bm{X}} \Big [ \frac{Y_i X_i}{2\sqrt{R(t,\epsilon)}} - \frac{X_i^2}{2} \Big ] \\
	& =-q(t,\epsilon) \sum_{i=1}^{n} \mathbb{E}_{X_i} \mathbb{E}_{Z_i}  \frac{Z_i X_i}{2 \sqrt{R(t,\epsilon)}} \\
	& = 0 
\end{align*}
by independence of the centered noise $\bm Z$ and the hidden partition $\bm X$.

Again the Nishimori identity \eqref{eq:Nishimori} is used to obtain
\begin{align*}
\mathbb{E} \Big\< \frac{d}{dt} \sH_{\mathrm{SBM}, t}\Big\>_{t,\epsilon}
	& = \frac{1}{2 \sqrt{1-t}} \sum_{i<j} \mathbb{E}  \Big\< \Delta_n x_i x_j \Big( \frac{ G_{ij} }{ \bar{p}_n + \sqrt{1-t} \Delta_n x_i x_j } - \frac{ 1-G_{ij} }{ 1 - \bar{p}_n - \sqrt{1-t} \Delta_n x_i x_j }\Big) \Big\>_{t,\epsilon} \\
	& = \frac{1}{2 \sqrt{1-t}} \sum_{i<j} \mathbb{E} \Big [ \Delta_n X_i X_j \Big ( \frac{ G_{ij} }{ \bar{p}_n + \sqrt{1-t} \Delta_n X_i X_j } - \frac{ 1-G_{ij} }{ 1 - \bar{p}_n - \sqrt{1-t} \Delta_n X_i X_j } \Big ) \Big ] \\
	& = \frac{1}{2 \sqrt{1-t}} \sum_{i<j} \mathbb{E}_{X_i, X_j} \Big [ \Delta_n X_i X_j \Big ( \frac{ \mathbb{E}_{G_{ij} | X_i, X_j} G_{ij} }{ \bar{p}_n + \sqrt{1-t} \Delta_n X_i X_j } - \frac{ 1-\mathbb{E}_{G_{ij} | X_i, X_j} G_{ij} }{ 1 - \bar{p}_n - \sqrt{1-t} \Delta_n X_i X_j } \Big ) \Big ] \\
	& = 0 \,,
\end{align*}
where the last line follows from $\mathbb{E}_{G_{ij} | X_i, X_j} G_{ij} = \bar{p}_n + \sqrt{1-t} \Delta_n X_i X_j$.
\end{proof}

\subsection{Final derivations of the sum rule.}

The last missing term in order to simplify the sum rule \eqref{eq:sum-rule:0} is:

\begin{lemma}\label{4.4}
We have $f_{0,0} - f_{0,\epsilon} = \frac{1}{2} \int_{0}^{\epsilon} d\epsilon'\, \mathbb{E} \< Q \>_{0,\epsilon'}$.
\end{lemma}
\begin{proof}
Using Gaussian integration by parts \eqref{eq:stein} and from \eqref{eq:Nishimori} the specific Nishimori identity $\mathbb{E} [\< x_i \>_{0,\epsilon'}^2]=\mathbb{E}[X_i\< x_i \>_{0,\epsilon'}]$ we have (recall also that $R(0,\epsilon')=\epsilon'$)
\begin{align*}
f_{0,0} - f_{0,\epsilon} &= - \int_{0}^{\epsilon} d\epsilon' \frac{df_{0,\epsilon'}}{d\epsilon'} = - \int_{0}^{\epsilon} d\epsilon' \Big\langle \frac{d}{d\epsilon'} \sH_{\mathrm{dec};t,\epsilon'}\Big\rangle_{0,\epsilon'} \\
	& = \int_{0}^{\epsilon} d\epsilon' \frac{1}{n} \sum_{i=1}^{n} \mathbb{E} \Big\< X_i x_i - \frac{x_i^2}{2} + \frac{1}{2\sqrt{\epsilon'}} Z_i x_i \Big\>_{0,\epsilon'} \\
	& = \int_{0}^{\epsilon} d\epsilon' \frac{1}{n} \sum_{i=1}^{n} \Big(\mathbb{E} \< X_i x_i \>_{0,\epsilon'} - \frac{1}{2} \mathbb{E} [\< x_i \>_{0,\epsilon'}^2]\Big) \\
	& = \frac{1}{2} \int_{0}^{\epsilon} d\epsilon' \,\mathbb{E} \< Q \>_{0,\epsilon'}\,,
\end{align*}
\end{proof}

Recall $R(1,\epsilon) = \epsilon + \int_0^1 q(t, \epsilon) dt$. Substituting \eqref{eq:chain-rule}, and Lemmas \ref{4.1}, \ref{4.2} and \ref{4.3} as well as \ref{4.4} into \eqref{eq:sum-rule:0} yields
\begin{align*}
\frac{1}{n} I(\bm{X}; \bm{G}) 
	& = \Psi(R(1,\epsilon), \lambda_n, r) - \frac{(\epsilon + \int_0^1 q(t,\epsilon) dt)^2}{4\lambda_n} + \frac{1}{2} \int_{0}^{\epsilon} d\epsilon' \,\mathbb{E} \< Q \>_{0,\epsilon'} \nonumber \\
	& \qquad- \int_0^1 dt \Big ( \frac{\lambda_n}{4} \mathbb{E} \< Q^2 \>_{t,\epsilon}  - \frac{1}{2} q(t,\epsilon)\mathbb{E} \< Q \>_{t,\epsilon}  \Big ) + o_n(1) \\
	& = \Psi(R(1,\epsilon), \lambda_n, r) + \frac{1}{4\lambda_n} \Big ( \int_0^1 q(t,\epsilon)^2 dt - \Big(\int_0^1 q(t,\epsilon) dt\Big)^2 \Big ) - \frac{1}{4\lambda_n} \int_0^1 dt \mathbb{E} \< ( \lambda_n Q - q(t,\epsilon) )^2 \>_{t,\epsilon} \nonumber \\
	& \qquad- \frac{\epsilon}{4\lambda_n} \Big(\epsilon + 2 \int_0^1 q(t,\epsilon) dt \Big) + \frac{1}{2} \int_0^{\epsilon} d\epsilon'\, \mathbb{E}\<Q\>_{0,\epsilon'} + o_n(1) 
\end{align*}
which is the sum rule \eqref{eq:sum-rule:1}.

\section{Concentration of overlap: proof of Lemma \ref{thm:overlap}}
\label{sec:overlap}
Concentration of overlap has been shown for various Bayesian inference problems, see, e.g., \cite{pmlr-BarKMMZ:2018,BarM:2018,2019arXiv190106516B}. These proofs can be adapted to the present case. The idea is to bound the fluctuations of the overlap by those of another, easier to control, object $\sL$ defined below. This object is more natural to work with as it is directly related to derivatives of the free energy, which, itself concentrates. Let us present the main steps of the proof, and then provide the proof details afterwards.

Let
\begin{align}
\sL \equiv \frac{1}{n} \sum_{i=1}^{n} \Big ( \frac{x_i^2}{2} - x_i X_i - \frac{x_i Z_i}{2 \sqrt{R(t,\epsilon)}} \Big )\,.
\label{eq:L}
\end{align}

As said previously, we can relate the fluctuations of the overlap to those of $\cal L$:
\begin{lemma}[A fluctuation identity] \label{eq:QL} We have $\mathbb{E} \< ( Q - \mathbb{E} \< Q \>_{t,\epsilon} )^2 \>_{t,\epsilon}\le 4\, \mathbb{E}\<(\sL -\mathbb{E}\< \sL \>_{t,\epsilon})^2 \>_{t,\epsilon}$.
\label{thm:QL}
\end{lemma}
%
It therefore remains to show the concentration of $\sL$. 
We divide the task into two parts:
\begin{align}
\mathbb{E} \< (\sL -\mathbb{E}\< \sL \>_{t,\epsilon})^2 \>_{t,\epsilon}  = \mathbb{E} \< ( \sL - \< \sL \>_{t,\epsilon} )^2 \>_{t,\epsilon} + \mathbb{E}[ ( \< \sL \>_{t,\epsilon} - \EE\< \sL \>_{t,\epsilon} )^2 ]\,. \label{eq:conc-L-decomp}
\end{align}

These two terms are controlled by the following lemmas:

\begin{lemma}[Thermal fluctuations] Let $R(t,\epsilon)=\epsilon+\int_0^t ds\, q(s,\epsilon) \ge \epsilon$ be such that $dR/d\epsilon\ge 1$. We then have $$\int_{s_n}^{2s_n} d\epsilon\, \mathbb{E} \< ( \sL - \< \sL \>_{t,\epsilon} )^2 \>_{t,\epsilon} \leq \frac{1}{n}\,.$$
\label{thm:thermal}
\end{lemma}

\begin{lemma}[Quenched fluctuations]
 Let $R(t,\epsilon)=\epsilon+\int_0^t ds\, q(s,\epsilon)$, with $\epsilon \in [s_n, 2s_n]$ and $q$ taking values in $[0,\lambda_n]$, be such that $dR/d\epsilon\ge 1$. There exists a sequence $C_n(r, \lambda_n) > 0$ converging to a constant such that
\begin{align}
\int_{s_n}^{2s_n} d\epsilon \,\mathbb{E}[ ( \< \sL \>_{t,\epsilon} - \EE\< \sL \>_{t,\epsilon} )^2 ] \leq \frac{C_n(r,\lambda_n)}{(s_n n)^{1/3}}\,.
\end{align}
\label{thm:disorder}
\end{lemma}

The proof of Lemma~\ref{thm:thermal} and Lemma~\ref{thm:disorder} employ some useful identities for the derivatives of the free energy (recall $F_{t,\epsilon} \equiv -\frac1n \ln{\cal Z}_{t,\epsilon}(\bm{G},\bm{Y})$):
\begin{align}
\frac{dF_{t,\epsilon}}{dR} 
	& = \< \sL \>_{t,\epsilon} \label{eq:F-der1}\,, \\
\frac{1}{n} \frac{d^2F_{t,\epsilon}}{dR^2}
	& = - (\< \sL^2 \>_{t,\epsilon} - \< \sL \>_{t,\epsilon}^2) + \frac{1}{4n^2 R^{3/2}} \sum_{i=1}^{n} \< x_i \>_{t,\epsilon} Z_i\,,
	\label{eq:F-der2}
\end{align}
where we simply denote, when no confusion can arise, $R=R(t,\epsilon)$. Taking expectation on both sides of \eqref{eq:F-der1} and \eqref{eq:F-der2} we have
\begin{align}
\frac{df_{t,\epsilon}}{dR}
	& = \mathbb{E} \< \sL \>_{t,\epsilon} = - \frac{1}{2n} \sum_{i=1}^{n} \mathbb{E}[\< x_i \>_{t,\epsilon}^2]\,, \label{eq:f-der1} \\
\frac{1}{n} \frac{d^2f_{t,\epsilon}}{dR^2}
	& =  - \mathbb{E} [\< \sL^2 \>_{t,\epsilon} - \< \sL \>_{t,\epsilon}^2] + \frac{1}{4n^2 R} \sum_{i=1}^{n} \mathbb{E} [ \< x_i^2 \>_{t,\epsilon} - \< x_i \>_{t,\epsilon}^2 ] \label{eq:f-der2} \\
	& = - \frac{1}{2n^2} \sum_{i,j=1}^{n} \mathbb{E}[(\< x_i x_j \>_{t,\epsilon} - \< x_i \>_{t,\epsilon} \< x_j \>_{t,\epsilon})^2]\,. \label{eq:f-der2:2}
\end{align}
The proof of Lemma \ref{thm:overlap} is ended by applying Lemmas \ref{eq:QL}, \ref{thm:thermal} and \ref{thm:disorder} in conjunction with \eqref{eq:conc-L-decomp}:
\begin{align*}
\frac{1}{s_n} \int_{s_n}^{2s_n} d\epsilon\, \mathbb{E} \< ( Q - \mathbb{E} \< Q \>_{t,\epsilon} )^2 \>_{t,\epsilon}  
	& \leq  \frac{4}{s_nn} + \frac{4C_n(r,\lambda_n)}{(s_n^4 n)^{1/3}} \,.
\end{align*}

We now provide the proofs of Lemmas~\ref{thm:QL} to \ref{thm:free-energy}. For the sake of readibility, we simply denote $\< - \> \equiv \< - \>_{t,\epsilon}$ for the rest of this section.

\subsection{Proof of Lemma \ref{thm:QL}}
We start by proving 
\begin{align}
-2\,\mathbb{E}\big\langle Q(\mathcal{L} - \mathbb{E}\langle \mathcal{L}\rangle)\big\rangle
&=\mathbb{E}\big\langle (Q - \mathbb{E}\langle Q \rangle)^2\big\rangle
+ \mathbb{E}\big\langle (Q-   \langle Q \rangle)^2\big\rangle\,.\label{47}
\end{align}
Using the definitions $Q \equiv \frac1n \sum_{i=1}^{n} x_i X_i$ and \eqref{eq:L} gives
\begin{align}
2\,\mathbb{E}\big\langle Q ({\cal L} -\mathbb{E}\langle {\cal L} \rangle) \big\rangle
	& = \frac{1}{n^2} \sum_{i,j=1}^{n} \Big\{\mathbb{E} \Big [  X_i \langle x_i x_j^2 \rangle - 2X_i X_j \langle x_i x_j \rangle - \frac{ Z_j}{\sqrt{R}} X_i \langle x_i x_j \rangle \Big ] \nonumber \\
	& \qquad\qquad \quad-  \mathbb{E} [X_i \langle x_i \rangle ] \, \mathbb{E} \Big [ \langle x_j^2 \rangle - 2X_j \langle x_j \rangle - \frac{ Z_j}{\sqrt{R}} \langle x_j \rangle \Big ]\Big\}\,. \label{eq:QL:1}
\end{align}
Gaussian integration by parts then yields
\begin{align*}
\mathbb{E}\Big[\frac{ Z_j}{\sqrt{R}} X_i \langle x_i x_j \rangle\Big] 
	& = \mathbb{E}[ X_i \langle x_i x_j^2 \rangle - X_i\langle x_i x_j \rangle \langle x_j \rangle ]\,, \quad\text{and} \quad \mathbb{E}\Big[\frac{ Z_j}{\sqrt{R}} \langle x_j \rangle\Big] 
	 =   \mathbb{E}[ \langle x_j^2 \rangle - \langle x_j \rangle^2 ]\,.
\end{align*}
These two formulas simplify \eqref{eq:QL:1} to
\begin{align}
&2\,\mathbb{E}\big\langle Q ({\cal L} -\mathbb{E}\langle {\cal L} \rangle ) \big\rangle\nonumber\\
	 &\qquad\qquad= \frac{1}{n^2} \sum_{i,j=1}^{n}\big\{ \mathbb{E} [ X_i \langle x_j \rangle \langle x_i x_j \rangle - 2X_i X_j \langle x_i x_j \rangle ] -  \mathbb{E} [X_i \langle x_i \rangle ] \,\mathbb{E} [ \langle x_j \rangle^2 - 2X_j \langle x_j \rangle ]\big\}\,.\label{toSimp}
\end{align}
The Nishimori identity implies
\begin{align*}
 \mathbb{E}[\langle x_j \rangle^2 ] = \mathbb{E}[X_j \langle x_j \rangle]\,, \quad \text{and}\quad \mathbb{E}[X_i \langle x_j \rangle \langle x_i x_j \rangle] = \mathbb{E}[ \langle x_i\rangle \langle x_j \rangle \langle x_i x_j \rangle ] = \mathbb{E}[ \langle x_i \rangle \langle x_j \rangle X_i X_j ]\,.
\end{align*}
These formulas further simplify \eqref{toSimp} to
\begin{align*}
2\,\mathbb{E}\big\langle Q ({\cal L} -\mathbb{E}\langle {\cal L} \rangle) \big\rangle & = \frac{1}{n^2} \sum_{i,j=1}^{n}\big\{ \mathbb{E} [  \langle x_i \rangle \langle x_j \rangle X_i X_j - 2X_i X_j \langle x_i x_j \rangle ] +  \mathbb{E} [X_i \langle x_i \rangle ] \,\mathbb{E} [ X_j \langle x_j \rangle ]\big\}\nonumber\\
	& = \mathbb{E}[\langle Q\rangle^2] - 2\,\mathbb{E}\langle Q^2\rangle +  \mathbb{E}[\langle Q\rangle]^2 \\
	& = -  \big ( \mathbb{E}\langle Q^2\rangle - \mathbb{E}[\langle Q\rangle]^2 \big ) -  \big ( \mathbb{E}\langle Q^2\rangle - \mathbb{E}[\langle Q\rangle^2] \big ) 
\end{align*}
which is \eqref{47}.

Identity \eqref{47} implies
\begin{align*}
2\big|\mathbb{E}\big\langle Q(\mathcal{L} - \mathbb{E}\langle \mathcal{L}\rangle)\big\rangle\big|=2\big|\mathbb{E}\big\langle (Q-\mathbb{E}\langle Q \rangle)(\mathcal{L} - \mathbb{E}\langle \mathcal{L}\rangle)\big\rangle\big|
\ge \mathbb{E}\big\langle (Q - \mathbb{E}\langle Q \rangle)^2\big\rangle
\end{align*}
and application of the Cauchy-Schwarz inequality then gives
\begin{align*}
2\big\{\mathbb{E}\big\langle (Q-\mathbb{E}\langle Q \rangle)^2\big\rangle\, \mathbb{E}\big\langle(\mathcal{L} - \mathbb{E}\langle \mathcal{L}\rangle)^2\big\rangle \big\}^{1/2}	\ge\mathbb{E}\big\langle (Q - \mathbb{E}\langle Q \rangle)^2\big\rangle\,.
\end{align*}
This ends the proof of Lemma \ref{thm:QL}.

\subsection{Proof of Lemma~\ref{thm:thermal}}
First note that $\frac{d^2 f_{t,\epsilon}}{ dR^2} \leq 0$. 
Then, using \eqref{eq:f-der2},  $dR/d\epsilon \geq 1$, $R(t,\epsilon) \geq \epsilon$, and the Nishimori identity $\EE\langle x_i^2\rangle = \EE[X_i^2]=1$, 
\begin{align*}
\mathbb{E} \< ( \sL - \< \sL \> )^2 \>
	& = - \frac{1}{n} \frac{d^2f_{t,\epsilon}}{dR^2} + \frac{1}{4n^2 R} \sum_{i=1}^{n} \mathbb{E} [ \< x_i^2 \> - \< x_i \>^2 ]\leq - \frac{1}{n} \frac{dR}{d\epsilon} \frac{d^2f_{t,\epsilon}}{dR^2} + \frac{1}{4n\epsilon}  = - \frac{1}{n} \frac{d}{d\epsilon} \Big(\frac{df_{ t,\epsilon}}{dR}\Big) + \frac{1}{4n\epsilon}\,,
\end{align*}
From \eqref{eq:f-der1} $df_{t,\epsilon}/dR \in [-1/2,0]$, therefore $ [ df_{t,\epsilon} / dR]_{\epsilon = s_n}^{\epsilon = 2s_n} \geq - 1/2$. Integrating over $\epsilon$ then gives 
\begin{align*}
\int_{s_n}^{2s_n} d\epsilon\, \mathbb{E} \< ( \sL - \< \sL \> )^2 \>
	& \leq \int_{s_n}^{2s_n} d\epsilon  \Big\{ - \frac{1}{n} \frac{d}{d\epsilon}\Big(\frac{df_ {t,\epsilon}}{dR}\Big) + \frac{1}{4n\epsilon} \Big\}  = - \frac{1}{n} \Big [ \frac{df_{t,\epsilon}}{dR} \Big ]_{\epsilon = s_n}^{\epsilon = 2s_n} + \frac{\ln 2}{4n}   \leq \frac{2+(\ln 2)}{4n} \leq \frac{1}{n}\,.
\end{align*}

\subsection{Proof of Lemma~\ref{thm:disorder}}
Lemma~\ref{thm:disorder} is based on the concentration of the free energy, a very general fact in "well behaved" statistical mechanics models. The proof of the following lemma uses more or less standard methods and can found in Appendix \ref{app-free-en}.
\begin{lemma}[Free energy fluctuations]
There exists a sequence $C_n(r, \lambda_n) > 0$ converging to a constant when $n\to +\infty$, such that
\begin{align}
{\rm Var}(F_{t,\epsilon})=\mathbb{E}[(F_{t,\epsilon} - f_{t,\epsilon})^2]  \leq \frac{C_n(r,\lambda_n)}{n}\,.
\end{align}
\label{thm:free-energy}
\end{lemma}

Recall $R=R(t,\epsilon)$. Let 
\begin{align}
& \tilde{F}_{t,\epsilon}(R) \equiv F_{t,\epsilon} + \sqrt{R\frac{1-r}{r}} \frac1n\sum_{i=1}^{n} | Z_i |, 
&& \tilde{f}_{t,\epsilon}(R) \equiv f_{t,\epsilon} +  \sqrt{R\frac{1-r}{r}} \frac1n\sum_{i=1}^{n} \mathbb{E} | Z_i | \label{eq:F-tilde}\,.
\end{align}
From \eqref{eq:f-der2:2} we see that $\tilde{f}_{t,\epsilon}(R)$ is concave in $R$. Furthermore, from \eqref{eq:F-der2} and $\vert x_i\vert \leq \sqrt{\frac{1-r}{r}}$ for $0\leq r\leq 1/2$, we see that $\tilde{F}_{t,\epsilon}(R)$ is also concave in $R$. So that we can employ the following lemma (see the end of this section for a proof):
\begin{lemma}[A bound on the difference of derivatives due to concavity]\label{lemma55}
Let $G(x)$ and $g(x)$ be concave functions. Let $\delta>0$ and define $C^{+}_\delta(x) \equiv g'(x) - g'(x+\delta) \geq 0$ and $C^{-}_\delta(x) \equiv g'(x-\delta) - g'(x) \geq 0$. Then
\begin{align*}
|G'(x) - g'(x)| \leq \delta^{-1} \sum_{u \in \{x-\delta, x, x+\delta\}} |G(u)-g(u)| + C^{+}_\delta(x) + C^{-}_\delta(x)\,.
\end{align*}
\label{thm:bound-by-convexity}
\end{lemma}

From \eqref{eq:F-tilde} we have
\begin{align*}
& \tilde{F}_{t,\epsilon} - \tilde{f}_{t,\epsilon} = F_{t,\epsilon} - f_{t,\epsilon} + \sqrt{R\frac{1-r}{r}} A_n\,, \qquad A_n \equiv \frac{1}{n} \sum_{i=1}^{n} ( | Z_i | - \mathbb{E}| Z_i | )\,,
\end{align*}
and from \eqref{eq:F-der1} and \eqref{eq:f-der1} we have
\begin{align*}
& \frac{d \tilde{F}_{t,\epsilon}}{dR} - \frac{d \tilde{f}_{t,\epsilon}}{dR}
	= \< \sL \> - \mathbb{E}\< \sL \> + \frac{1}{2} \sqrt{\frac{1-r}{Rr}} A_n\,.
\end{align*}
Using Lemma \ref{thm:bound-by-convexity} we then get
\begin{align*}
 \big | \< \sL \> - \mathbb{E} \< \sL \> \big | &\leq \delta^{-1} \sum_{u \in \{R-\delta, R, R + \delta\}} \big ( |F_{t,\epsilon}(R=u) - f_{t,\epsilon}(R=u)| +  \sqrt{u\frac{1-r}{r}} |A_n| \big ) \\
&\qquad\qquad\qquad\qquad+ C_{\delta}^{+}(R) + C_{\delta}^{-}(R) + \frac{1}{2} \sqrt{\frac{1-r}{Rr}} A_n
\end{align*}
where $C_{\delta}^{+}(R) \equiv \tilde{f}_{t,\epsilon}'(R) - \tilde{f}_{t,\epsilon}'(R+\delta) \geq 0$ and $C_{\delta}^{-}(R) \equiv \tilde{f}'_{t,\epsilon}(R-\delta) - \tilde{f}'_{t,\epsilon}(R) \geq 0$. Then squaring this inequality, using $(\sum_{i=1}^{p} v_i)^2 \leq p \sum_{i=1}^{p} v_i^2$, taking the expectation, and recalling that $R=R(t,\epsilon) \geq \epsilon$ we reach
\begin{align}
\frac{1}{9} \mathbb{E}\big [ (\< \sL \> - \mathbb{E} \< \sL \>)^2 \big ]
	& \leq \delta^{-2} \sum_{u \in \{R-\delta, R, R + \delta\}} \Big\{ \mathbb{E}[(F_{t,\epsilon}(u) - f_{t,\epsilon}(u))^2] + u \frac{1-r}{r} \mathbb{E}[A_n^2] \Big\} + C_{\delta}^{+}(R)^2 + C_{\delta}^{-}(R)^2 \nonumber \\
	& \qquad\qquad\qquad\qquad\qquad\qquad+  \frac{1-r}{4 \epsilon r} \mathbb{E}[A_n^2]\,. \label{eq:disorder:1}
\end{align}
Note that $\mathbb{E}[A_n^2] = a/n$ with $a=1-2/\pi$. Recall $q^*(t,\epsilon) \in [0, \lambda_n]$ from Lemma~\ref{thm:m-star}. We can upper bound $u$ by $\lambda_n +2s_n+\delta$. These remarks with Lemma~\ref{thm:free-energy} simplify \eqref{eq:disorder:1} to
\begin{align}
\frac{1}{9} \mathbb{E}\big [ (\< \sL \> - \mathbb{E} \< \sL \>)^2 \big ]
	\leq \frac{3}{n\delta^2} \bigg ( C_n(r,\lambda_n) + a ( \lambda_n+2s_n+\delta )\frac{1-r}{r}\bigg ) + C_{\delta}^{+}(R)^2 + C_{\delta}^{-}(R)^2 + \frac{1}{4 \epsilon} \frac{1-r}{r} \frac{a}{n}\,. \label{eq:disorder:2}
\end{align}
Recall \eqref{eq:f-der1} and that $\EE[\langle x_i\rangle^2] \le \EE\langle x_i^2\rangle=\EE[X_i^2]=1$. We have $$|\tilde{f}'_{t,\epsilon}(R)| \leq \frac{1}{2}\Big(1+\sqrt{\frac{1-r}{rR}}\Big)$$ and therefore $0 \leq C_{\delta}^{\pm}(R) \leq 1+\sqrt{\frac{1-r}{r(R-\delta)}}$. Using $dR /d\epsilon \geq 1$ and $R\ge s_n$ we then have
\begin{align*}
\int_{s_n}^{2s_n} d\epsilon \big\{C_{\delta}^{+}(R)^2 + C_{\delta}^{-}(R)^2\big\}
	& \leq 2 \Big(1+\sqrt{\frac{1-r}{r(s_n-\delta)}}\Big) \int_{s_n}^{2s_n} d\epsilon \big\{C_{\delta}^{+}(R) + C_{\delta}^{-}(R)\big\} \\
	& = 2 \Big(1+\sqrt{\frac{1-r}{r(s_n-\delta)}}\Big) \int_{s_n}^{2s_n} d\epsilon\Big( \frac{d\tilde{f}_{t,\epsilon}(R-\delta)}{dR} - \frac{d\tilde{f}_{t,\epsilon}(R+\delta)}{dR} \Big ) \\
	& \leq 2 \Big(1+\sqrt{\frac{1-r}{r(s_n-\delta)}}\Big) \int_{s_n}^{2s_n} d\epsilon \frac{dR}{d\epsilon} \Big ( \frac{d\tilde{f}_{t,\epsilon}(R-\delta)}{dR} - \frac{d\tilde{f}_{t,\epsilon}(R+\delta)}{dR} \Big ) \\
	& = 2 \Big(1+\sqrt{\frac{1-r}{r(s_n-\delta)}}\Big) \int_{s_n}^{2s_n} d\epsilon \Big ( \frac{d\tilde{f}_{t,\epsilon}(R(t,\epsilon)-\delta)}{d\epsilon} - \frac{d\tilde{f}_{t,\epsilon}(R(t,\epsilon)+\delta)}{d\epsilon} \Big ) \\
	& =2 \Big(1+\sqrt{\frac{1-r}{r(s_n-\delta)}}\Big) \big \{ \big ( \tilde{f}_{t,2s_n}(R(t,2s_n)-\delta) - \tilde{f}_{t,2s_n}(R(t,2s_n)+\delta) \big ) \nonumber \\
	& \qquad\qquad+ \big ( \tilde{f}_{t,s_n}(R(t,s_n)+\delta) -  \tilde{f}_{t,s_n}(R(t,s_n)-\delta) \big ) \big \} \\
	& \leq 4 \delta \Big(1+\sqrt{\frac{1-r}{r(s_n-\delta)}}\Big)^2
\end{align*}
using the mean value theorem for the last step. Therefore upon integrating \eqref{eq:disorder:2} over $\epsilon \in (s_n, 2s_n)$ we have
\begin{align}
\frac{1}{9} \int_{s_n}^{2s_n} d\epsilon\, \mathbb{E}\big [ (\< \sL \> - \mathbb{E} \< \sL \>)^2 \big ]
	& \leq \frac{3s_n}{n\delta^2} \big ( C_n(r,\lambda_n) + a (\lambda_n  +2s_n+\delta)  \frac{1-r}{r}  \big ) \nonumber \\
	& + 4 \delta \Big(1+\sqrt{\frac{1-r}{r(s_n-\delta)}}\Big)^2 + \frac{a(1-r)\ln 2}{4rn}\,.
\end{align}
The bound is optimized choosing $\delta = (s_n^2/n)^{1/3}$. This ends the proof.

\begin{proof}[Proof of Lemma \ref{lemma55}]
Concavity implies that for any $\delta>0$ we have
\begin{align*}
G'(x) - g'(x)
 	& \geq \frac{G(x+\delta) - G(x)}{\delta} - g'(x) \nonumber \\
	& \geq \frac{G(x+\delta) - G(x)}{\delta} - g'(x) + g'(x+\delta) - \frac{g(x+\delta) - g(x)}{\delta} \nonumber \\
	& = \frac{G(x+\delta) - g(x+\delta)}{\delta} - \frac{G(x) - g(x)}{\delta} - C^+_\delta(x) \,,  \\
G'(x) - g'(x)
	& \leq \frac{G(x) - G(x-\delta)}{\delta} - g'(x) + g'(x-\delta) - \frac{g(x) - g(x-\delta)}{\delta} \nonumber \\
	& = \frac{G(x) - g(x)}{\delta} - \frac{G(x-\delta) - g(x-\delta)}{\delta} + C^-_{\delta}(x) \,.
\end{align*}
Combining these two inequalities ends the proof.
\end{proof}

\section{Approximate integration by parts: proof of lemma \ref{lemm:intparts}}
\label{sec:approxint}

The following general formula follows from Taylor expansion with Lagrange remainder. When the r.h.s is small in specific applications, the formula can be seen as an approximate integration by parts formula generalizing 
Gaussian integration by parts. 

\begin{lemma}
Let $g(U)$ be a $\mathcal{C}^4$ function of a random variable $U$ such that for $k=1,2,3,4$ we have  $\sup_{U} \big | g^{(k)}(U) \big | \leq C_k$ for some constants $C_k \geq 0$ and $g^{(k)}(U) \equiv d^{k}g(U) / dU^k$. Suppose that the first four moments of $U$ are finite. Then
\begin{align}
	& \Big | \mathbb{E}_{}[Ug(U)] - \mathbb{E}_{}[g'(U)] \mathbb{E}_{}[U^2] - g(0) \mathbb{E}_{}U \Big | \nonumber \\
	& \qquad\qquad\leq C_{2} \bigg( \frac{ \big| \mathbb{E}_{} [U^3]  \big |}{2} + \mathbb{E}[U^2]\mathbb{E}U  \bigg) + C_{3} \bigg ( \frac{\mathbb{E}_{} [ U^4 ]}{24} + \frac{\mathbb{E}_{}[U^2]^2}{2} \bigg ) + \frac{C_{4}}{6} \big| \mathbb{E}_{} [U^3] \big | \mathbb{E}_{} [U^2]\,. \label{appIPP}
\end{align}
\label{thm:stein}
\end{lemma}
\begin{proof}
By Taylor's theorem any $\mathcal{C}^4$ function $h(U)$ can be written as
\begin{align*}
h(U) = h(0) + h^{(1)}(0) U + \frac{1}{2} h^{(2)}(0) U^2 + \frac{1}{2} \int_0^U h^{(3)}(s) (U - s)^2 ds\,.
\end{align*}
Taking the expectation on both sides:
\begin{align}
\mathbb{E}h(U) = h(0) + h^{(1)}(0) \mathbb{E}_{}U + \frac{1}{2} h^{(2)}(0) \mathbb{E}_{} [ U^2 ] + \frac{1}{2} \mathbb{E}  \int_0^U h^{(3)}(s) (U - s)^2 ds \,.
\label{eq:taylor-theorem}
\end{align}
When \eqref{eq:taylor-theorem} is applied to $h(U) = g^{(1)}(U)$ we have
\begin{align}
\mathbb{E}_{}g^{(1)}(U) & = g^{(1)}(0) + g^{(2)}(0) \mathbb{E}U + \frac{1}{2} g^{(3)}(0) \mathbb{E}_{}[U^2] + \frac{1}{2} \mathbb{E} \int_0^U g^{(4)}(s) (U-s)^2 ds \,.
\label{eq:stein-expansion-1}
\end{align}
On the other hand when \eqref{eq:taylor-theorem} is applied to $h(U) = U g(U)$, using $(U g(U))^{(k)} = U g^{(k)}(U) + k g^{(k-1)}(U)$ we have
\begin{align}
\mathbb{E}[Ug(U)] - g(0) \mathbb{E}_{}U = g^{(1)}(0) \mathbb{E}_{}[U^2] + \frac{1}{2} \mathbb{E}_{} \int_0^U (s g^{(3)}(s) + 3 g^{(2)}(s) ) (U-s)^2 ds \,.
\label{eq:stein-expansion-2}
\end{align}
Subtracting \eqref{eq:stein-expansion-1} and \eqref{eq:stein-expansion-2} we have the bound
\begin{align}
	 \Big | \mathbb{E}_{}[Ug(U)] - \mathbb{E}_{}&g^{(1)}(U) \mathbb{E}_{}[U^2] - g(0) \mathbb{E}_{}U \Big | \nonumber \\
	& = \Big | \frac{1}{2} \mathbb{E}_{}  \int_0^U (s g^{(3)}(s) + 3 g^{(2)}(s) ) (U-s)^2 ds  - g^{(2)}(0)  \mathbb{E}_{}[U^2]\mathbb{E}_{}U - \frac{1}{2}g^{(3)}(0) \mathbb{E}_{}[U^2]^2 \nonumber \\
	& \qquad\qquad- \frac{1}{2} \mathbb{E}_{}[U^2]\mathbb{E}_{}  \int_0^U g^{(4)}(s) (U-s)^2 ds \Big | \nonumber \\
	& \leq \frac{C_{3}}{2} \Big | \mathbb{E}_{} \int_0^U s (U-s)^2 ds \Big | + \frac{3C_{2}}{2} \Big | \mathbb{E}_{} \int_0^U (U-s)^2 ds \Big | + C_{2} \mathbb{E}_{}[U^2]\mathbb{E}_{}U  \nonumber \\
	& \qquad\qquad+ \frac{C_{3}}{2} \mathbb{E}_{}[U^2]^2 + \frac{C_{4}}{2} \mathbb{E}_{} [U^2]  \Big | \mathbb{E}_{} \int_0^U (U-s)^2 ds \Big | \nonumber \\
	& = \frac{C_{3}}{24} \mathbb{E}_{} [ U^4 ] + \frac{C_{2}}{2} \big | \mathbb{E}_{} [U^3] \big | + C_{2} \mathbb{E}_{}U \mathbb{E}_{}[U^2] + \frac{C_{3}}{2} \mathbb{E}_{}[U^2]^2 + \frac{C_{4}}{6} \big | \mathbb{E}_{} [U^3] \big | \mathbb{E}_{} [U^2] \,,
	\label{eq:stein-diff}
\end{align}
which is the right hand side of \eqref{appIPP} after factorization.
\end{proof}

We now apply lemma \ref{thm:stein} to our specific problem in order to derive the approximate integration by parts formula \eqref{eq:talagrand}.

\vskip 0.25cm 

\noindent{\it Proof of lemma \ref{lemm:intparts}.}
In order to apply lemma \ref{thm:stein} to the SBM, consider $U = G_{ij}$ and $g(U) = F_{t,\epsilon}(G_{ij})$ the free energy \eqref{free-en} seen as a function of $G_{ij}$ (all other variables being fixed).
For the expectation we take $\mathbb{E}=\mathbb{E}_{G_{ij} | {X}_i, {X}_j}$.
At time $t$ and for any integer $k$
\begin{align*}
\mathbb{E}_{G_{ij} | X_i, X_j} [G_{ij}^k]
	 =\mathbb{E}_{G_{ij} | X_i, X_j} G_{ij} = \bar{p}_n + \sqrt{1-t} \Delta_n X_i X_j = \mathcal{O}(\bar{p}_n),
\end{align*}
because $G_{ij}\in \{0,1\}$. 
For the derivatives we note that using the Taylor expansion of the logarithm, one obtains for any $v_n \in \mathbb{R}$ and $v_n \rightarrow 0$, $\ln (1+v_n) - v_n = \mathcal{O}(|v_n|^2)$, which also implies $\ln(1+v_n) = \mathcal{O}(|v_n|)$. (The reader should keep this fact in mind, as it is used again in the appendices whenever we need to expand the logarithm.)
Now this fact implies
\begin{align*}
- F_{t,\epsilon}^{(1)}(G_{ij}) 
	& = \frac{1}{n} \Big\< \ln (1 + \frac{\Delta_n}{\bar{p}_n} \sqrt{1-t} x_i x_j) - \ln ( 1- \frac{\Delta_n}{1-\bar{p}_n} \sqrt{1-t} x_i x_j ) \Big\>_{t,\epsilon} \\
	& = \frac{1}{n} \bigg ( \frac{\Delta_n}{\bar{p}_n} + \frac{\Delta_n}{1-\bar{p}_n} \bigg ) \sqrt{1-t} \< x_i x_j \>_{t,\epsilon} + \mathcal{O} \Big( \frac{1}{n} \Big ( \frac{\Delta_{n}}{1-\bar{p}_n}\Big )^2 (1-t) \Big) + \mathcal{O} \Big( \frac{1}{n} \Big ( \frac{\Delta_{n}}{(1-\bar{p}_n)}\Big )^2 (1-t) \Big) \\
	& = \frac{1}{n} \frac{\Delta_n}{\bar{p}_n(1-\bar{p}_n)} \sqrt{1-t} \< x_i x_j \>_{t,\epsilon} + \mathcal{O} \Big( \frac{1}{n} \Big ( \frac{\Delta_{n}}{\bar{p}_n(1-\bar{p}_n)}\Big )^2 (1-t) \Big)\,,\\
- F_{t,\epsilon}^{(2)}(G_{ij}) 
	& = \frac{1}{n} \Big\< \Big ( \ln (1 + \frac{\Delta_n}{\bar{p}_n} \sqrt{1-t} x_i x_j) - \ln ( 1- \frac{\Delta_n}{1-\bar{p}_n} \sqrt{1-t} x_i x_j ) \Big )^2 \Big\>_{t,\epsilon} \nonumber \\
	& \qquad- \frac{1}{n}  \Big\< \ln (1 + \frac{\Delta_n}{\bar{p}_n} \sqrt{1-t} x_i x_j) - \ln ( 1- \frac{\Delta_n}{1-\bar{p}_n} \sqrt{1-t} x_i x_j ) \Big\>_{t,\epsilon}^2 \\
	& = \mathcal{O}\Big( \frac{1}{n} \Big ( \frac{\Delta_{n}}{\bar{p}_n(1-\bar{p}_n)} \Big )^2 (1-t)\Big)\,.
\end{align*}
To obtain these identities the reader has again to be careful in performing the derivatives: both the exponential of the Hamiltonian and the partition function appearing in the definition of the Gibbs-bracket depend on $(G_{ij})$ (see the derivation of \eqref{last} for similar computations).
In general, $$| F_{t,\epsilon}^{(k)}(G_{ij}) | =  \mathcal{O} \Big(\frac{1}{n} \Big ( \frac{\Delta_{n}}{\bar{p}_n(1-\bar{p}_n)} \Big )^{k} (1-t)^{k/2} \Big)\,.$$
Using Lemma~\ref{thm:stein} we have
\begin{align*}
	A_n & \equiv \Big | \mathbb{E}_{G_{ij} | X_i, X_j}[G_{ij} F_{t,\epsilon}] + \mathbb{E}_{G_{ij} | X_i, X_j} [G_{ij}]\\
	&\qquad \times\Big\{ \frac{1}{n} \frac{\Delta_n}{\bar{p}_n(1-\bar{p}_n)} \sqrt{1-t} \mathbb{E}_{G_{ij} | X_i, X_j} [\< x_i x_j \>_{t,\epsilon}] + \mathcal{O} \Big( \frac{1-t}{n} \Big ( \frac{\Delta_{n}}{\bar{p}_n(1-\bar{p}_n)} \Big )^2 \Big) - F_{t,\epsilon}(G_{ij}=0) \Big\} \Big | \nonumber \\
	& = \mathcal{O} \Big ( \frac{\sqrt{1-t}}{n} \Big ( ( \frac{\Delta_n}{\bar{p}_n (1-\bar{p}_n) } )^2 (\bar{p}_n + \bar{p}_n^2 ) + ( \frac{\Delta_n}{\bar{p}_n (1-\bar{p}_n) } )^3 (\bar{p}_n + \bar{p}_n^2 ) + ( \frac{\Delta_n}{\bar{p}_n (1-\bar{p}_n) } )^4 \bar{p}_n^2 \Big ) \Big ) \\
	& = \mathcal{O} \Big ( \frac{\sqrt{1-t}}{n} \frac{\Delta_n^2}{\bar{p}_n(1-\bar{p}_n)^2} \Big )\,.
\end{align*}
Then by the triangle inequality we extract
\begin{align*}
	 \Big | \mathbb{E}_{G_{ij} | X_i, X_j}[G_{ij} &F_{t,\epsilon}] 
	 + \mathbb{E}_{G_{ij} | X_i, X_j} [G_{ij}] \Big\{ \frac{1}{n} \frac{\Delta_n}{\bar{p}_n} \sqrt{1-t} \mathbb{E}_{G_{ij} | X_i, X_j} [\< x_i x_j \>_{t,\epsilon}]  
	 - F_{t,\epsilon}(G_{ij}=0) \Big\} \Big | \nonumber \\
	& \leq A_n + (\bar{p}_n + \sqrt{1-t} \Delta_n X_i X_j) \mathcal{O} \Big( \frac{1-t}{n} \Big( \frac{\Delta_{n}}{\bar{p}_n (1-\bar{p}_n)} \Big )^2  \Big) \\
	& = \mathcal{O} \Big ( \frac{\sqrt{1-t}}{n} \frac{\Delta_n^2}{\bar{p}_n (1-\bar{p}_n)^2} \Big ) \\
	&= \mathcal{O} \Big ( \frac{\sqrt{1-t} \lambda_n}{n^2(1-\bar{p}_n)} \Big )\,.
\end{align*}
and recognize formula \eqref{eq:talagrand}.
\appendix
\section{Mutual information and free energy: proof of Proposition~\ref{thm:MI}}
\label{app:MI}

Using \eqref{eq:transition-prob}, we have the expression
\begin{align*}
I(\bm{X}; \bm{G}) 
	& \equiv \mathbb{E}_{\bm{X}} \mathbb{E}_{\bm{G} | \bm{X}} \ln \biggl\{\frac{\mathbb{P}(\bm{G}|\bm{X})}{\mathbb{P}(\bm{G})}\biggr\} = \mathbb{E}_{\bm{X}} \mathbb{E}_{\bm{G} | \bm{X}} \ln \biggl\{\frac{\mathbb{P}(\bm{G}|\bm{X})}{\sum_{\bm{x} \in \sX^n } \mathbb{P}_r(\bm{x}) \mathbb{P}(\bm{G} | \bm{x} )}\biggr\} \\
	& = \mathbb{E}_{\bm{X}} \mathbb{E}_{\bm{G} | \bm{X}} \ln \biggl\{\frac{ \prod_{i<j} (\bar{p}_n+\Delta_n X_i X_j)^{G_{ij}}  (1-\bar{p}_n-\Delta_n X_i X_j)^{1-G_{ij}} }{ \sum_{\bm{x} \in \sX^n} \mathbb{P}_{r}(\bm{x}) \prod_{i<j} (\bar{p}_n+\Delta_n x_i x_j)^{G_{ij}}  (1-\bar{p}_n-\Delta_n x_i x_j)^{1-G_{ij}} }\biggr\}\,.
\end{align*}
We divide both the numerator and denominator by the same factor, and then rewrite the denominator in exponential form: 
\begin{align}
I(\bm{X}; \bm{G}) 
	& = \mathbb{E}_{\bm{X}} \mathbb{E}_{\bm{G} | \bm{X}} \ln \biggl\{\frac{ \prod_{i<j} (1+\frac{\Delta_n}{\bar{p}_n} X_i X_j)^{G_{ij}}  (1-\frac{\Delta_n}{1-\bar{p}_n} X_i X_j)^{1-G_{ij}} }{ \sum_{\bm{x} \in \sX^n} \mathbb{P}_r(\bm{x}) \prod_{i<j} (1+\frac{\Delta_n}{\bar{p}_n} x_i x_j)^{G_{ij}}  (1-\frac{\Delta_n}{1-\bar{p}_n} x_i x_j)^{1-G_{ij}} }\biggr\} \nonumber \\
	& = \mathbb{E}_{\bm{X}} \mathbb{E}_{\bm{G} | \bm{X}} \ln \biggl\{\frac{ \prod_{i<j} (1+\frac{\Delta_n}{\bar{p}_n} X_i X_j)^{G_{ij}}  (1-\frac{\Delta_n}{1-\bar{p}_n} X_i X_j)^{1-G_{ij}} }{ \sum_{\bm{x} \in \sX^n} \mathbb{P}_r(\bm{x}) \exp \sum_{i<j} \big ( G_{ij} \ln (1+\frac{\Delta_n}{\bar{p}_n} x_i x_j) + (1-G_{ij}) \ln (1-\frac{\Delta_n}{1-\bar{p}_n} x_i x_j) \big )  }\biggr\} \nonumber \\
	& = \mathbb{E}_{\bm{X}} \mathbb{E}_{\bm{G} | \bm{X}} \ln \biggl\{\prod_{i<j} (1+\frac{\Delta_n}{\bar{p}_n} X_i X_j)^{G_{ij}}  (1-\frac{\Delta_n}{1-\bar{p}_n} X_i X_j)^{1-G_{ij}} \bigg \} - \mathbb{E}_{\bm{X}} \mathbb{E}_{\bm{G} | \bm{X}} \ln {\cal Z}(\bm{G}). \label{eq:mi1}
\end{align}
Recall $\mathbb{E}_{G_{ij} | X_i, X_j} G_{ij} = \bar{p}_n + \Delta_n X_i X_j$. The first term in \eqref{eq:mi1} equals
\begin{align}
& \sum_{i<j} \mathbb{E}_{\bm{X}} \mathbb{E}_{\bm{G} | \bm{X}} \bigg \{ G_{ij} \ln (1+\frac{\Delta_n}{\bar{p}_n} X_i X_j) + (1-G_{ij}) \ln (1-\frac{\Delta_n}{1-\bar{p}_n} X_i X_j) \bigg \} \nonumber \\
& = \sum_{i<j} \mathbb{E}_{\bm{X}} \bigg \{ (\bar{p}_n+\Delta_n X_i X_j) \ln ( 1 + \frac{\Delta_n}{\bar{p}_n} X_i X_j ) +  (1 - \bar{p}_n - \Delta_n X_i X_j) \ln (1-\frac{\Delta_n}{1-\bar{p}_n} X_i X_j) \bigg \}. \label{eq:mi-no}
\end{align}
Let $X\sim \mathbb{P}_r$. We can further write explicitly the expectation in \eqref{eq:mi-no} that leads us to conclude
\begin{align}
& \frac{1}{n} I(\bm{X}; \bm{G}) = \frac{n-1}{2} \bigg \{ r^2 (\bar{p}_n + \Delta_n \frac{1-r}{r}) \ln ( 1 + \frac{\Delta_n}{\bar{p}_n} \frac{1-r}{r} ) + r^2(1-\bar{p}_n-\Delta_n \frac{1-r}{r}) \ln (1-\frac{\Delta_n}{1-\bar{p}_n} \frac{1-r}{r} ) \nonumber \\
	& + (1-r)^2 (\bar{p}_n + \Delta_n \frac{r}{1-r}) \ln ( 1 + \frac{\Delta_n}{\bar{p}_n} \frac{r}{1-r} ) + (1-r)^2(1-\bar{p}_n-\Delta_n \frac{r}{1-r}) \ln (1-\frac{\Delta_n}{1-\bar{p}_n} \frac{r}{1-r} ) \nonumber \\
	& + 2r(1-r) (\bar{p}_n - \Delta_n) \ln ( 1 - \frac{\Delta_n}{\bar{p}_n} ) + 2r(1-r)(1 - \bar{p}_n + \Delta_n ) \ln ( 1 + \frac{\Delta_n}{1-\bar{p}_n} ) \bigg \} - \frac{1}{n} \mathbb{E}_{\bm{X}} \mathbb{E}_{\bm{G} | \bm{X}} \ln {\cal Z}(\bm{G}). \label{eq:mi-exact}
\end{align}
Using the Taylor expansion of the logarithm, \eqref{eq:mi-exact} becomes
\begin{align*}
& \frac{1}{n} I(\bm{X}; \bm{G}) = \frac{\lambda_n(n-1)}{4n} - \frac{1}{n} \mathbb{E}_{\bm{X}} \mathbb{E}_{\bm{G} | \bm{X}} \ln {\cal Z}(\bm{G}) + \frac{n-1}{2} \sum_{k=3}^{\infty} \frac{\Delta_n^k}{k(k-1)} \big ( \frac{1}{\bar{p}_n^{k-1}} + \frac{(-1)^k}{(1-\bar{p}_n)^{k-1}} \big ) \mathbb{E}[X^k]^2,
\end{align*}
where $\mathbb{E}[X^k]^2 = r^2 (\frac{1-r}{r})^{k} + (1-r)^2 (\frac{r}{1-r})^{k} + (-1)^{k} 2 r (1-r)$.
This becomes the expression in \eqref{eq:MI} by noting that the last term is $\mathcal{O}\Big ( n\Delta_n^3/\big (\bar{p}_n (1-\bar{p}_n)\big )^2 \Big ) = \mathcal{O}(\lambda_n^{3/2} / \sqrt{n \bar{p}_n (1-\bar{p}_n)})$.


%
\section{Liouville formula}
\label{app:Liouville}
Consider the differential equation \eqref{eq:m-star} with $G_n(t,R(t,\epsilon)) = \lambda_n \mathbb{E} \< Q \>_{t,\epsilon}$. Differentiating w.r.t $\epsilon$ and using 
the chain rule gives
\begin{align*}
\frac{d}{dt} \frac{dR}{d\epsilon}(t,\epsilon)
	= \frac{dR}{d\epsilon}(t,\epsilon) \frac{dG_n}{dR}(t,R(t,\epsilon))\,.
\end{align*}
Therefore we have
\begin{align}
\frac{d}{dt} \ln \biggl\{\frac{dR}{d\epsilon}(t,\epsilon)\biggr\} = \frac{dG_n}{dR}(t,R(t,\epsilon)) \,.\label{eq:Lio1}
\end{align}
Integrating \eqref{eq:Lio1} over $t \in [0,t']$ we have
\begin{align}
\ln\biggl\{\frac{dR}{d\epsilon}(t',\epsilon)\biggr\}  - \ln\biggl\{\frac{dR}{d\epsilon}(0,\epsilon)\biggr\}  = \int_0^{t'} dt \frac{dG_n}{dR}(t,R(t,\epsilon))\,. \label{eq:Lio2}
\end{align}
Using $R(0,\epsilon)=\epsilon$, \eqref{eq:Lio2} implies
\begin{align}
\frac{dR}{d\epsilon}(t',\epsilon) = \exp  \biggl\{\int_0^{t'}dt \frac{dG_n}{dR}(t, R(t,\epsilon))\biggr\} \,.
\end{align}
This is known as Liouville's formula for one-dimensional ordinary differential equations.
\section{Small error terms in the sum rule: proof of \eqref{eq:comp4:1b}}
\label{app:D1}

Recalling the definitions \eqref{eq:H-SBM-t} and \eqref{eq:H-Y-t}, let
\begin{align}
\sH_{t,\epsilon}(\bm{x}; \bm{G} \setminus G_{ij}, \bm{Y}) & \equiv \sH_{\mathrm{SBM};t}(\bm{x}; \bm{G} \setminus G_{ij}) + \sH_{\mathrm{dec};t,\epsilon}(\bm{x}; \bm{Y})\,, \label{eq:H-remove-one} \\
\sH_{\mathrm{SBM};t}(\bm{x}; \bm{G} \setminus G_{ij} ) & \equiv - \sum_{k<l: (k,l) \notin \{(i,j), (j,i)\}} \Big \{ G_{kl} \ln (1 + \frac{\Delta_n}{\bar{p}_n} \sqrt{1-t} x_k x_l) \nonumber \\
& \qquad\qquad \qquad\qquad\qquad+ (1-G_{kl}) \ln ( 1- \frac{\Delta_n}{1-\bar{p}_n} \sqrt{1-t} x_k x_l ) \Big \}\,. \nonumber
\end{align}
Also let $F_{t,\epsilon; \sim G_{ij}} \equiv n^{-1} \ln \sum_{{\bm{x}} \in \sX^n } e^{-\sH_{t,\epsilon}({\bm{x}};{\bm{G} \setminus G_{ij}},{\bm{Y}})} \mathbb{P}_r(\bm{x})$, and $\<-\>_{t,\epsilon; \sim G_{ij}}$ be the Gibbs-bracket associated to the measure proportional to $\sH_{t,\epsilon}(\bm{x}; \bm{G} \setminus G_{ij}, \bm{Y})$.
The difference of free energy when changing one $G_{ij}$ can be written in terms of this Gibbs-bracket:
\begin{align}
&\mathbb{E}_{G_{ij} | X_i, X_j} F_{t,\epsilon} - F_{t,\epsilon}(G_{ij}=0)
	= \mathbb{P}_t(G_{ij} = 1 | X_i, X_j) ( F_{t,\epsilon}(G_{ij}=1) - F_{t,\epsilon}(G_{ij}=0) ) \label{eq:energy-diff-Gij-1} \\
	&\qquad = \mathbb{P}_t(G_{ij} = 1 | X_i, X_j) \{ ( F_{t,\epsilon}(G_{ij}=1) - F_{t,\epsilon; \sim G_{ij}} ) - ( F_{t,\epsilon}(G_{ij}=0) -F_{t,\epsilon; \sim G_{ij}} ) \} \nonumber \\
	&\qquad= -(\bar{p}_n + \sqrt{1-t} \Delta_n X_i X_j)\frac1n \bigg \{ \ln \frac{\sum_{{\bm{x}} \in \sX^n } e^{-\sH_{t,\epsilon}({\bm{x}};{\bm{G} \setminus G_{ij}},{\bm{Y}})+(\sH_{t,\epsilon}({\bm{x}};{\bm{G} \setminus G_{ij}},{\bm{Y}})-\sH_{t,\epsilon}({\bm{x}};{\bm{G}},G_{ij}=1,{\bm{Y}}))} \mathbb{P}_r(\bm{x})}{\sum_{{\bm{x}} \in \sX^n } e^{-\sH_{t,\epsilon}({\bm{x}};{\bm{G} \setminus G_{ij}},{\bm{Y}})} \mathbb{P}_r(\bm{x})} \nonumber \\
	&\qquad \qquad - \ln \frac{\sum_{{\bm{x}} \in \sX^n } e^{-\sH_{t,\epsilon}({\bm{x}};{\bm{G} \setminus G_{ij}},{\bm{Y}})+(\sH_{t,\epsilon}({\bm{x}};{\bm{G} \setminus G_{ij}},{\bm{Y}})-\sH_{t,\epsilon}({\bm{x}};{\bm{G}},G_{ij}=0,{\bm{Y}}))} \mathbb{P}_r(\bm{x})}{\sum_{{\bm{x}} \in \sX^n } e^{-\sH_{t,\epsilon}({\bm{x}};{\bm{G} \setminus G_{ij}},{\bm{Y}})} \mathbb{P}_r(\bm{x})} \bigg \} \nonumber \\
	&\qquad = - (\bar{p}_n + \sqrt{1-t} \Delta_n X_i X_j)  \frac{1}{n} \bigg \{ \ln \Big\< e^{\sH_{\mathrm{SBM};t}({\bm{x}};{\bm{G}}\setminus G_{ij})-\sH_{\mathrm{SBM};t}({\bm{x}};{\bm{G}},G_{ij}=1)} \Big\>_{t,\epsilon;\sim G_{ij}} \nonumber \\
	&\qquad \qquad - \ln \Big\< e^{\sH_{\mathrm{SBM};t}({\bm{x}};{\bm{G}}\setminus G_{ij})-\sH_{\mathrm{SBM};t}({\bm{x}};{\bm{G}},G_{ij}=0)} \Big\>_{t,\epsilon;\sim G_{ij}} \bigg \} \nonumber \\
	&\qquad = - (\bar{p}_n + \sqrt{1-t} \Delta_n X_i X_j)  \frac{1}{n} \bigg \{ \ln \< 1 + \frac{\Delta_n}{\bar{p}_n} \sqrt{1-t} x_i x_j \>_{t,\epsilon;\sim G_{ij}} - \ln \< 1 - \frac{\Delta_n}{1-\bar{p}_n} \sqrt{1-t} x_i x_j \>_{t,\epsilon;\sim G_{ij}} \bigg \} \,. \label{eq:energy-diff-Gij-2}
\end{align}
Using the Taylor expansion of the logarithms in \eqref{eq:energy-diff-Gij-2}, 
we have
\begin{align*}
\mathbb{E}_{G_{ij} | X_i, X_j} F_{t,\epsilon} - F_{t,\epsilon}(G_{ij}=0)
	& = - (\bar{p}_n + \sqrt{1-t} \Delta_n X_i X_j) \frac{1}{n} \bigg \{ \frac{\Delta_n \sqrt{1-t}}{\bar{p}_n(1-\bar{p}_n)}  \< x_i x_j \>_{t,\epsilon; \sim G_{ij}} \\
	& \qquad + \sum_{k=2}^{\infty} \frac{\Delta_n^k}{k} \bigg ( \frac{(-1)^k}{\bar{p}_n^k} - \frac{1}{(1-\bar{p}_n)^k} \bigg ) (1-t)^{k/2} \< x_i x_j \>_{t,\epsilon; \sim G_{ij}}^k \bigg \} \\
	& = - (\bar{p}_n + \sqrt{1-t} \Delta_n X_i X_j) \frac{\Delta_n \sqrt{1-t}}{n \bar{p}_n(1-\bar{p}_n)}  \< x_i x_j \>_{t,\epsilon; \sim G_{ij}} + \mathcal{O} \Big ( \frac{\Delta_n^2 (1-t)}{n \bar{p}_n (1-\bar{p}_n)^2} \Big ).
\end{align*}
Therefore, replacing in the expression of $E_1$, we find
\begin{align*}
E_1 = E_1^{(a)} + E_1^{(b)}
\end{align*}
where
\begin{align}
& E_1^{(a)} = \frac{\Delta_n^2}{2n \bar{p}_n(1-\bar{p}_n)} \sum_{i<j} \mathbb{E}_{\sim G_{ij}} \Big [ \frac{\bar{p}_n + \sqrt{1-t} \Delta_n X_i X_j}{1 - \bar{p}_n - \sqrt{1-t} \Delta_n X_i X_j} X_i X_j \< x_i x_j \>_{t,\epsilon; \sim G_{ij}} \Big ], \nonumber \\
& E_1^{(b)} 
	= \mathcal{O} \Big ( \frac{\Delta_n n^2}{\sqrt{1-t} (1-\bar{p}_n)} \cdot \frac{\Delta_n^2 (1-t)}{n \bar{p}_n (1-\bar{p}_n)^2} \Big ) 
	= \mathcal{O} \Big ( \frac{n \Delta_n^3}{\bar{p}_n (1-\bar{p}_n)^3} \Big ) 
	= \mathcal{O}\Big ( \frac{\lambda_n \Delta_n}{(1-\bar{p}_n)^2} \Big ). \label{eq:E1b}
\end{align}
We then observe that
\begin{align}
E_1^{(a)} + E_2 = \frac{\Delta_n^2}{2n \bar{p}_n(1-\bar{p}_n)} \sum_{i<j} \mathbb{E}_{\sim G_{ij}} \Big [ \frac{\bar{p}_n + \sqrt{1-t} \Delta_n X_i X_j}{1 - \bar{p}_n - \sqrt{1-t} \Delta_n X_i X_j} X_i X_j \big ( \mathbb{E}_{G_{ij} | X_i, X_j}[\< x_i x_j \>_{t,\epsilon}] - \< x_i x_j \>_{t,\epsilon; \sim G_{ij}} \big )\Big ]. \label{eq:comp4:1b:small2}
\end{align}
The difference between the Gibbs-brackets in \eqref{eq:comp4:1b:small2} can be expanded as
\begin{align}
\mathbb{E}_{G_{ij} | X_i, X_j}[\< x_i x_j \>_{t,\epsilon}] - \< x_i x_j \>_{t,\epsilon; \sim G_{ij}} 
	& = \mathbb{P}_{t}(G_{ij}=1|X_i,X_j) (\< x_i x_j \>_{t,\epsilon; G_{ij=1}} - \< x_i x_j \>_{t,\epsilon; \sim G_{ij}}) \nonumber \\
	& \quad + \mathbb{P}_{t}(G_{ij}=0|X_i,X_j) (\< x_i x_j \>_{t,\epsilon; G_{ij=0}} - \< x_i x_j \>_{t,\epsilon; \sim G_{ij}}), \label{eq:comp4:1b:gibbs-diff}
\end{align}
and we can evaluate $\< x_i x_j \>_{t,\epsilon; G_{ij=1}} - \< x_i x_j \>_{t,\epsilon; \sim G_{ij}}$ by an interpolation:
\begin{align}
	& \< x_i x_j \>_{t,\epsilon; G_{ij=1}} - \< x_i x_j \>_{t,\epsilon; \sim G_{ij}} \nonumber \\
	& \quad = \int_0^1 ds \frac{d}{ds} \bigg \{ \frac{\sum_{{\bm{x}} \in \sX^n } x_i x_j \exp \big \{-\sH_{t,\epsilon}({\bm{x}};{\bm{G} \setminus G_{ij}},{\bm{Y}}) + s \ln (1 + x_i x_j \sqrt{1-t} \frac{\Delta_n}{\bar{p}_n} ) \big \} \mathbb{P}_r(\bm{x}) }{ \sum_{{\bm{x}} \in \sX^n } \exp \big \{ -\sH_{t,\epsilon}({\bm{x}};{\bm{G} \setminus G_{ij}},{\bm{Y}}) + s \ln (1 + x_i x_j \sqrt{1-t} \frac{\Delta_n}{\bar{p}_n} ) } \mathbb{P}_r(\bm{x}) \big \} \bigg \} \nonumber \\
	& \quad = \int_0^1 ds \bigg \{ \< x_i x_j \ln ( 1 + x_i x_j \sqrt{1-t} \frac{\Delta_n}{\bar{p}_n} ) \>_{t,\epsilon; s} - \< x_i x_j \>_{t,\epsilon; s} \< \ln ( 1 + x_i x_j \sqrt{1-t} \frac{\Delta_n}{\bar{p}_n} ) \>_{t,\epsilon; s} \bigg \}, \label{eq:comp4:interpolation}
\end{align}
where $\< - \>_{t,\epsilon;s}$ is the Gibbs-bracket associated to the measure proportional to 
$$\exp \big \{ -\sH_{t,\epsilon}({\bm{x}};{\bm{G} \setminus G_{ij}},{\bm{Y}}) + s \ln (1 + x_i x_j \sqrt{1-t} \frac{\Delta_n}{\bar{p}_n} ) \big \} \mathbb{P}_r(\bm{x})$$ 
with $\sH_{t,\epsilon}({\bm{x}};{\bm{G} \setminus G_{ij}},{\bm{Y}})$ defined in \eqref{eq:H-remove-one}.
By the Taylor expansion of the logarithms in \eqref{eq:comp4:interpolation} and using $\mathbb{P}_{t}(G_{ij}=1|X_i,X_j) = \mathcal{O}(\bar{p}_n)$, we see that the first term of \eqref{eq:comp4:1b:gibbs-diff} is $\mathcal{O}(\Delta_n)$. The same kind of calculation is used to see that the second term of \eqref{eq:comp4:1b:gibbs-diff} is also $\mathcal{O}(\Delta_n)$. This implies for \eqref{eq:comp4:1b:small2}
\begin{align}
E_1^{(a)} + E_2 = \mathcal{O}\Big ( \frac{n\Delta_n^3}{(1-\bar{p}_n)^2}\Big ) = \mathcal{O}\Big ( \frac{\lambda_n \Delta_n \bar{p}_n}{(1-\bar{p}_n)^2} \Big ), \label{eq:comp4:E1a-E2}
\end{align}
which tends to zero.
Now we conclude by noting that $E_1+E_2=E_1^{(a)}+E_1^{(b)}+E_2$ and using \eqref{eq:E1b} and \eqref{eq:comp4:E1a-E2} to obtain \eqref{eq:comp4:1b}.
\section{Concentration of free energy: proof of Lemma~\ref{thm:free-energy}}
\label{app-free-en}

The generation of quenched variables can be divided into two stages: firstly $\bm{X}$, then $\bm{G}$ given $\bm{X}$, and independently the Gaussian noise $\bm{Z}$. We expand the variance of free energy according to the two stages (recall $f_{t,\epsilon}=\mathbb{E}_{\bm X}\mathbb{E}_{\bm{G} | \bm{X}} \mathbb{E}_{\bm{Z}}F_{t,\epsilon}$):
\begin{align}
\mathbb{E}[(F_{t,\epsilon} - f_{t,\epsilon})^2] 
	& = \mathbb{E}[( F_{t,\epsilon} - \mathbb{E}_{\bm{G} | \bm{X}} \mathbb{E}_{\bm{Z}}F_{t,\epsilon})^2] + \mathbb{E}[(\mathbb{E}_{\bm{G} | \bm{X}} \mathbb{E}_{\bm{Z}}F_{t,\epsilon}-f_{t,\epsilon})^2]\,.  \label{eq:conc-F-decompose}
\end{align}
In each stage the variables are all independently generated. This enables us to use Efron-Stein inequality to show the concentration of free energy.

Let $\bm{Z}^{(i)}$ be a vector such that $\bm{Z}^{(i)}$ differs from $\bm{Z}$ only at the $i$-th which becomes $Z_i'$ drawn independently from the same distribution as the one of $Z_i\sim {\cal N}(0,1)$. We define $\bm{G}^{(ij)}$ and $\bm{X}^{(i)}$ in the similar manner with respect to $\bm{G}$ and $\bm{X}$. 
Efron-Stein's inequality tells us that
\begin{align}
\mathbb{E}[( F_{t,\epsilon} - \mathbb{E}_{\bm{G} | \bm{X}} \mathbb{E}_{\bm{Z}}F_{t,\epsilon})^2]
	& \leq \frac{1}{2} \sum_{i=1}^{n} \mathbb{E}_{\bm X}\mathbb{E}_{\bm{G} | \bm{X}} \mathbb{E}_{\bm{Z}} \mathbb{E}_{Z_i'}[(F_{t,\epsilon}(\bm{Z}) - F_{t,\epsilon}(\bm{Z}^{(i)}))^2] \nonumber \\
	& \qquad+ \frac{1}{2} \sum_{i<j}\mathbb{E}_{\bm X} \mathbb{E}_{\bm{G} | \bm{X}} \mathbb{E}_{G'_{ij} | \bm{X}} \mathbb{E}_{\bm{Z}} [(F_{t,\epsilon}(\bm{G}) - F_{t,\epsilon}(\bm{G}^{(ij)}))^2]\,, \label{eq:EF-GZ}
\end{align}
as well as
\begin{align}
\mathbb{E}[(\mathbb{E}_{\bm{G} | \bm{X}} \mathbb{E}_{\bm{Z}}F_{t,\epsilon}-f_{t,\epsilon})^2]
	& \leq \frac{1}{2} \sum_{i=1}^{n} \mathbb{E}_{\bm{X}} \mathbb{E}_{X'_i} [ ( \mathbb{E}_{\bm{G} | \bm{X}} \mathbb{E}_{\bm{Z}}F_{t,\epsilon}(\bm{X}) - \mathbb{E}_{\bm{G} | \bm{X}^{(i)}} \mathbb{E}_{\bm{Z}}F_{t,\epsilon}(\bm{X}^{(i)}) )^2]\,. \label{eq:EF-X}
\end{align}
By \eqref{eq:conc-F-decompose} it suffices to show that both \eqref{eq:EF-GZ} and \eqref{eq:EF-X} are upper bounded by $ C_n(r, \lambda_n) / n$ for some large enough sequence $C_n(r, \lambda_n)$ that converges to a constant. 

\subsection{Bound on \eqref{eq:EF-GZ}}

The bound obtained from Efron-Stein's inequality is a sum of local variances of the free energy. The bound on the difference due to a local change can be estimated by interpolation. For the first one we have
\begin{align*}
| F_{t,\epsilon}(\bm{Z}) - F_{t,\epsilon}(\bm{Z}^{(i)}) |
	& = \frac{1}{n} \Big| \int_0^1 ds \frac{d}{ds} \ln \sum_{\bm{x} \in \sX^n} \exp \big \{ -s\sH_{t,\epsilon}(\bm{x}; \bm{G},\bm{X},\bm{Z}) - (1-s)\sH_{t,\epsilon}(\bm{x}; \bm{G},\bm{X},\bm{Z}^{(i)}) \big \} \mathbb{P}_{r}(\bm{x}) \Big| \\
	& = \frac{1}{n} \Big| \int_0^1 ds \< \sH_{\mathrm{dec}; t,\epsilon}(\bm{x}; \bm{X},\bm{Z}^{(i)}) - \sH_{\mathrm{dec}; t,\epsilon}(\bm{x}; \bm{X},\bm{Z}) \>_{s} \Big| \\
	& =  \frac{1}{n} \Big| \int_0^1 ds \sqrt{R(t,\epsilon)} \< x_i \>_{s} (Z'_i - Z_i) \Big| \\
	& \leq \frac{1}{n}  \sqrt{(2s_n + \lambda_n)\frac{1-r}{r}} |Z'_i - Z_i| 
\end{align*}
where the Gibbs-bracket $\langle - \rangle_s$ is associated to the measure proportional to $\exp \{ -s\sH_{t,\epsilon}(\bm{x}; \bm{G},\bm{X},\bm{Z}) - (1-s)\sH_{t,\epsilon}(\bm{x}; \bm{G},\bm{X},\bm{Z}^{(i)}) \}$.
This implies an upper bound on the first sum in \eqref{eq:EF-GZ}:
\begin{align*}
	& \frac{1}{2} \sum_{i=1}^{n} \mathbb{E}_{\bm{G} | \bm{X}} \mathbb{E}_{\bm{Z}} \mathbb{E}_{Z_i'}[(F_{t,\epsilon}(\bm{Z}) - F_{t,\epsilon}(\bm{Z}^{(i)}))^2] \leq \frac{1}{2n^2} ( 2s_n + \lambda_n) \frac{1-r}{r}  \sum_{i=1}^{n} \mathbb{E}[(Z'_i-Z_i)^2] \leq \frac{C_n(r,\lambda_n)}{n}\,.
\end{align*}

Another interpolation gives
\begin{align*}
| F_{t,\epsilon}(\bm{G}) - F_{t,\epsilon}&(\bm{G}^{(ij)}) |\\
	& = \frac{1}{n} \Big | \int_0^1 ds \frac{d}{ds} \ln \sum_{\bm{x} \in \sX^n} \exp \big \{ -s\sH_{t,\epsilon}(\bm{x}; \bm{G},\bm{X},\bm{Z}) - (1-s)\sH_{t,\epsilon}(\bm{x}; \bm{G}^{(ij)},\bm{X},\bm{Z}) \big \} \mathbb{P}_{r}(\bm{x}) \Big | \\
	& = \frac{1}{n} \Big | (G'_{ij}-G_{ij}) \Big\< \ln (1 + \frac{\Delta_n}{\bar{p}_n} \sqrt{1-t} \tX_i \tX_j ) - \ln (1 - \frac{\Delta_n}{1-\bar{p}_n} \sqrt{1-t} \tX_i \tX_j \Big\>_{s} \Big | \\
	& \leq \frac{C(r) \Delta_n}{2n \bar{p}_n (1-\bar{p}_n)} | G'_{ij}-G_{ij} |
\end{align*}
for some constant $C(r)$, and where $\langle - \rangle_s$ is associated to the measure proportional to $\exp \{ -s\sH_{t,\epsilon}(\bm{x}; \bm{G},\bm{X},\bm{Z}) - (1-s)\sH_{t,\epsilon}(\bm{x}; \bm{G}^{(ij)},\bm{X},\bm{Z}) \}$. This bounds the second sum in \eqref{eq:EF-GZ} as
\begin{align*}
	\frac{C(r) \Delta_n^2}{2n^2 \bar{p}_n^2 (1-\bar{p}_n)^2} \sum_{i<j} \mathbb{E}_{G_{ij} | X_i, X_j} \mathbb{E}_{G'_{ij} | X_i, X_j}[(G'_{ij}-G_{ij})^2]
	= \frac{C(r) \Delta_n^2}{n^2 \bar{p}_n^2 (1-\bar{p}_n)^2} \sum_{i<j} \mathrm{Var}_{G_{ij} | X_i, X_j} (G_{ij})
	\leq \frac{C_n(r, \lambda_n)}{n},
\end{align*}
using that $(G_{ij})$ are $0$, $1$ Bernoulli variables, and the variance $$\mathrm{Var}_{G_{ij} | X_i, X_j}(G_{ij}) = (\bar{p}_n + \Delta_n \sqrt{1-t} X_i X_j) (1 - \bar{p}_n + \Delta_n \sqrt{1-t} X_i X_j)$$ as well as $\big ( \Delta_n/\big ( p_n(1-\bar{p}_n) \big ) \big )^2 = \lambda_n/(n\bar{p}_n (1-\bar{p}_n))$ in the last inequality.

\subsection{Bound on \eqref{eq:EF-X}}

We relax \eqref{eq:EF-X} with inequality $((a-c)+(c-b))^2 \leq 2(a-c)^2 + 2(c-b)^2$ so that 
\begin{align}
\mathbb{E}[(\mathbb{E}_{\bm{G} | \bm{X}} \mathbb{E}_{\bm{Z}}F_{t,\epsilon}-f_{t,\epsilon})^2]	& \leq \sum_{i=1}^{n} \mathbb{E}_{\bm{X} } \mathbb{E}_{X_i'} [ ( \mathbb{E}_{\bm{G} | \bm{X}} \mathbb{E}_{\bm{Z}}F_{t,\epsilon}(\bm{X}) - \mathbb{E}_{\bm{G} | \bm{X}} \mathbb{E}_{\bm{Z}}F_{t,\epsilon}(\bm{X}^{(i)}) )^2] \nonumber \\
	& \qquad+ \sum_{i=1}^{n} \mathbb{E}_{\bm{X} } \mathbb{E}_{X_i'} [ ( \mathbb{E}_{\bm{G} | \bm{X}} \mathbb{E}_{\bm{Z}}F_{t,\epsilon}(\bm{X}^{(i)}) - \mathbb{E}_{\bm{G} | \bm{X}^{(i)}} \mathbb{E}_{\bm{Z}}F_{t,\epsilon}(\bm{X}^{(i)}) )^2]\,. \label{eq:EF-X:2-decomp}
\end{align}
The difference in the first sum is given by
\begin{align*}
	& | F_{t,\epsilon}(\bm{X}) - F_{t,\epsilon}(\bm{X}^{(i)}) | \nonumber \\
	& \qquad= \frac{1}{n} \Big | \int_0^1 ds \frac{d}{ds} \ln \sum_{\bm{x} \in \sX^n} \exp \big \{ -s\sH_{t,\epsilon}(\bm{G},\bm{X},\bm{Z},\tbX) - (1-s)\sH_{t,\epsilon}(\bm{G},\bm{X}^{(i)},\bm{Z},\tbX) \big \} \mathbb{P}_{r}(\bm{x}) \\
	& \qquad= \frac{1}{n} \Big | \int_0^1 ds\, R(t,\epsilon) \< x_i \>_{s} (X'_i - X_i) \Big |
\end{align*}
where $\langle - \rangle_s$ is associated to the measure proportional to $\exp \{ -s\sH_{t,\epsilon}(\bm{G},\bm{X},\bm{Z},\tbX) - (1-s)\sH_{t,\epsilon}(\bm{G},\bm{X}^{(i)},\bm{Z},\tbX) \}$. Therefore the sum of square is bounded by $C_n(r, \lambda_n) / n$ using $R(t,\epsilon) \in [0, \lambda_n]$.

For the second sum we use another interpolation:
\begin{align}
\mathbb{E}_{\bm{G} | \bm{X}} \mathbb{E}_{\bm{Z}}F_{t,\epsilon}(\bm{X}^{(i)}) - \mathbb{E}_{\bm{G} | \bm{X}^{(i)}} \mathbb{E}_{\bm{Z}}F_{t,\epsilon}(\bm{X}^{(i)})
	& = \int_0^1 ds \sum_{\bm{G}} \frac{d}{ds} \mathbb{P}_{t,s}(\bm{G} | \bm{X}, X'_i) \mathbb{E}_{\bm{Z}}F_{t,\epsilon}(\bm{X}^{(i)})\,, \label{eq:EF-X:2a}
\end{align}
where
\begin{align*}
\mathbb{P}_{t,s}(\bm{G} | \bm{X},X'_i) 
	\equiv &\prod_{j:j \neq i}^{n} (\bar{p}_n \!+\! \sqrt{1\!-\!t} \Delta_n ((1\!-\!s)X_i+sX'_i) X_j )^{G_{ij}} (1\! -\! \bar{p}_n - \sqrt{1\!-\!t} \Delta_n ((1\!-\!s)X_i\!+\!sX'_i) X_j )^{1-G_{ij}} \nonumber \\
	&\times\prod_{\substack{k<l: \\ k,l\neq i}} (\bar{p}_n \!+\! \sqrt{1\!-\!t} \Delta_n X_k X_l)^{G_{kl}} (1\!-\! \bar{p}_n \!-\!  \sqrt{1\!-\!t} \Delta_n X_k X_l)^{1-{G_{kl}}}\,.
\end{align*}
As $G_{ij} \in \{0,1\}$, we have various ways to write $\mathbb{P}_{t,s}(\bm{G} | \bm{X},X'_i)$. A convenient way is using
\begin{align*}
P_{ij} & \equiv (\bar{p}_n \!+\! \sqrt{1\!-\!t} \Delta_n ((1\!-\!s)X_i+sX'_i) X_j )^{G_{ij}} (1\! -\! \bar{p}_n - \sqrt{1\!-\!t} \Delta_n ((1\!-\!s)X_i\!+\!sX'_i) X_j )^{1-G_{ij}} \\
& = G_{ij} \{ \bar{p}_n \!+\! \sqrt{1\!-\!t} \Delta_n ((1\!-\!s)X_i+sX'_i) X_j \} + (1-G_{ij}) \{ 1\! -\! \bar{p}_n - \sqrt{1\!-\!t} \Delta_n ((1\!-\!s)X_i\!+\!sX'_i) X_j \}\,.
\end{align*}
A compact formula for $dP_{ij} / ds$ can then be derived:
\begin{align}
\frac{dP_{ij}}{ds} = (2G_{ij}-1)  \sqrt{1-t} \Delta_n (X_i' - X_i) X_j  = (-1)^{1+G_{ij}} \sqrt{1-t} \Delta_n (X_i' - X_i) X_j\,.
\label{eq:EF-X:2b}
\end{align}
Let $\bm{G}_{\sim (i,j)} \equiv \bm{G} \setminus G_{ij}$ and $\mathbb{P}_{t,s}(\bm{G}_{\sim (i,j)} | \bm{X}, X'_i) \equiv \sum_{G_{ij} \in \{0,1\}} \mathbb{P}_{t,s}(\bm{G} | \bm{X}, X'_i) $ be the marginal of this sub-graph. Using \eqref{eq:EF-X:2b} we obtain
\begin{align}
\frac{d}{ds} \mathbb{P}_{t,s}(\bm{G} | \bm{X}, X'_i) 
	& = \sum_{j:j \neq i}^{n} \frac{dP_{ij}}{ds} \mathbb{P}_{t,s}(\bm{G}_{\sim (i,j)} | \bm{X}, X'_i) \nonumber \\
	& = \sum_{{j:j \neq i}}^{n} \sqrt{1-t} \Delta_n (X_i' - X_i) X_j (-1)^{1+G_{ij}} \mathbb{P}_{t,s}(\bm{G}_{\sim (i,j)} | \bm{X}, X'_i)\,.
	\label{eq:EF-X:2c}
\end{align}
Substituting~\eqref{eq:EF-X:2c} into \eqref{eq:EF-X:2a} gives
\begin{align}
	 \int_0^1 &ds  \sum_{\bm{G}} \sum_{j:j \neq i}^{n} \sqrt{1-t} \Delta_n (X_i' - X_i) X_j (-1)^{1+G_{ij}} \mathbb{P}_{t,s}(\bm{G}_{\sim (i,j)} | \bm{X}, X'_i) \mathbb{E}_{\bm{Z}} F_{t,\epsilon}(\bm{X}^{(i)})\nonumber \\
	& = \int_0^1 ds  \sum_{j:j \neq i}^{n} \sqrt{1-t} \Delta_n (X_i' - X_i) X_j \sum_{G_{ij} \in \{0,1\}}(-1)^{1+G_{ij}} \mathbb{E}_{\bm{G}_{\sim (i,j)} | \bm{X}, X'_i} \mathbb{E}_{\bm{Z}} F_{t,\epsilon}(\bm{X}^{(i)})  \nonumber \\
	& = \int_0^1 ds  \sum_{j:j \neq i}^{n}  \sqrt{1-t} \Delta_n (X_i' - X_i) X_j  \mathbb{E}_{\bm{G}_{\sim (i,j)} | \bm{X}, X'_i} \mathbb{E}_{\bm{Z}} [F_{t,\epsilon}(\bm{X}^{(i)}, G_{ij} = 1) - F_{t,\epsilon}(\bm{X}^{(i)}, G_{ij} = 0)] \,, \label{eq:EF-X:2d}
\end{align}
where $\mathbb{E}_{\bm{G}_{\sim (i,j)} | \bm{X}, X'_i}$ corresponds to the expectation with respect to the distribution $\mathbb{P}_{t,s}(\bm{G}_{\sim (i,j)} | \bm{X}, X'_i)$.
To evaluate the difference of free energy in \eqref{eq:EF-X:2d}, first we define $\bm{Y}^{(i)}=\sqrt{R(t,\epsilon)}\bm{X}^{(i)}+\bm{Z}$, and $\< - \>_{t,\epsilon;\bm{X}^{(i)}, \sim G_{ij}}$ is associated to $\exp\{- \sH_{t,\epsilon}(\bm{x}; \bm{G} \setminus G_{ij}, \bm{Y}^{(i)}) \}$ defined in \eqref{eq:H-remove-one}.
The same calculation as in \eqref{eq:energy-diff-Gij-1} -- \eqref{eq:energy-diff-Gij-2} gives
\begin{align}
F_{t,\epsilon}(&\bm{X}^{(i)}, G_{ij} = 1) - F_{t,\epsilon}(\bm{X}^{(i)}, G_{ij} = 0) \nonumber \\
& =  - \frac{1}{n} \Big \{ \ln \< 1 + \frac{\Delta_n}{\bar{p}_n} \sqrt{1-t} x_i x_j \>_{t,\epsilon;\bm{X}^{(i)}, \sim G_{ij}} - \ln \< 1 - \frac{\Delta_n}{1-\bar{p}_n} \sqrt{1-t} x_i x_j \Big \}. \label{eq:EF-X:2e}
\end{align}
Expanding the logarithms we can see \eqref{eq:EF-X:2e} is $\mathcal{O} \big ( \Delta_n / (n\bar{p}_n(1-\bar{p}_n)) \big )$. Using this fact and that all other terms inside the sum of \eqref{eq:EF-X:2d} are upper bounded by constants, we see that \eqref{eq:EF-X:2d} is $\mathcal{O}\big (\Delta_n^2 / (\bar{p}_n(1-\bar{p}_n) \big ) = \mathcal{O}(\lambda_n / n)$.
We can then upper bound the second term of \eqref{eq:EF-X:2-decomp}:
\begin{align*}
\sum_{i=1}^{n} \mathbb{E}_{\bm{X}} \mathbb{E}_{X'_i} \big [ ( \mathbb{E}_{\bm{G} | \bm{X}} \mathbb{E}_{\bm{Z}}F_{t,\epsilon}(\bm{X}^{(i)}) - \mathbb{E}_{\bm{G} | \bm{X}^{(i)}} \mathbb{E}_{\bm{Z}}F_{t,\epsilon}(\bm{X}^{(i)}) )^2 \big ] \leq \frac{C_n(r, \lambda_n)}{n}\,.
\end{align*}
\section*{Acknowledgments}
This work was supported by the SNSF grant no. 200021-156672.

	\bibliographystyle{IEEEtran}


\end{document}